\documentclass[10pt,letterpaper]{article}

\usepackage[T1]{fontenc}
\usepackage{lmodern}
\usepackage{microtype}
\usepackage[margin=1in]{geometry}
\usepackage[hyphens]{url}
\usepackage{graphicx}
\usepackage[authoryear,round]{natbib}
\let\cite\citep
\usepackage{caption}
\usepackage{algorithm}
\usepackage{algorithmic}
\usepackage{tikz}
\usetikzlibrary{arrows.meta}
\usepackage{booktabs}
\usepackage{amssymb}
\usepackage{amsmath}
\usepackage{amsthm}
\usepackage{bm}
\usepackage{xcolor}
\usepackage{soul}
\usepackage{placeins}
\usepackage[hidelinks]{hyperref}
\usepackage[nameinlink]{cleveref}
\crefname{assumption}{assumption}{assumptions}
\Crefname{assumption}{Assumption}{Assumptions}
\hypersetup{
    pdftitle={Estimating Causal Bounds from Observational Data With Neural Generative Models},
    pdfauthor={Max Zhu, Srivathsan Amruth, Sina Akbari, Qingyuan Zhao}
}
\newtheorem{theorem}{Theorem}
\newtheorem{proposition}{Proposition}
\newtheorem{assumption}{Assumption}
\newtheorem{corollary}{Corollary}
\newcommand{\pa}[1]{\operatorname{Pa}_{#1}}
\newcommand{\E}{\mathbb{E}}
\newcommand{\g}{\mathcal{G}}

\title{Estimating Causal Bounds from Observational Data With Neural Generative Models}
\author{%
    Max Zhu\thanks{Corresponding author: \href{mailto:mz406@cam.ac.uk}{mz406@cam.ac.uk}.} \quad Srivathsan Amruth \quad Sina Akbari \quad Qingyuan Zhao\\[0.5em]
    \small Department of Pure Mathematics and Mathematical Statistics\\
    \small University of Cambridge
}
\date{}

\begin{document}

\maketitle

\begin{abstract}
Identifying causal bounds from discrete observational data is a fundamental problem in causal inference. While exact analytical bounds exist for certain graphs, this is not generally the case. In this paper, we introduce a framework for estimating such bounds using neural generative models together with a new extrapolation procedure to estimate these bounds more accurately. Across multiple causal inference settings, it produces bounds that are consistently closer to the true bounds than those obtained with existing general neural approaches. We validate the framework on a range of benchmark problems for which exact bounds are available.
\end{abstract}

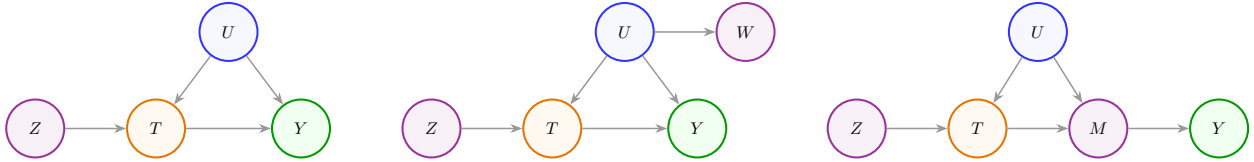
\begin{figure*}[t]
    \centering
    \resizebox{\textwidth}{!}{\begin{tikzpicture}[
    >=Stealth,
    arrow/.style={
        ->,
        draw=gray!80,
        line width=0.9pt
    },
    treatment/.style={
        circle,
        draw=orange!90!black,
        fill=orange!5,
        line width=1.2pt,
        minimum size=12mm
    },
    latent/.style={
        circle,
        draw=blue!80,
        fill=blue!3,
        line width=1.2pt,
        minimum size=12mm
    },
    outcome/.style={
        circle,
        draw=green!60!black,
        fill=green!6,
        line width=1.2pt,
        minimum size=12mm
    },
    observed/.style={
        circle,
        draw=violet!80,
        fill=violet!6,
        line width=1.2pt,
        minimum size=12mm
    }
]

\node[observed]  (z-iv) at (0,0)   {$Z$};
\node[treatment] (t-iv) at (2.5,0) {$T$};
\node[outcome]   (y-iv) at (5.5,0) {$Y$};
\node[latent]    (u-iv) at (4,2)   {$U$};
\draw[arrow] (z-iv) -- (t-iv);
\draw[arrow] (u-iv) -- (t-iv);
\draw[arrow] (u-iv) -- (y-iv);
\draw[arrow] (t-iv) -- (y-iv);

\begin{scope}[xshift=8.2cm]
    \node[observed]  (z-proxy) at (0,0)   {$Z$};
    \node[treatment] (t-proxy) at (2.5,0) {$T$};
    \node[outcome]   (y-proxy) at (5.5,0) {$Y$};
    \node[latent]    (u-proxy) at (4,2)   {$U$};
    \node[observed]  (w-proxy) at (6.5,2) {$W$};
    \draw[arrow] (z-proxy) -- (t-proxy);
    \draw[arrow] (t-proxy) -- (y-proxy);
    \draw[arrow] (u-proxy) -- (t-proxy);
    \draw[arrow] (u-proxy) -- (y-proxy);
    \draw[arrow] (u-proxy) -- (w-proxy);
\end{scope}

\begin{scope}[xshift=17cm]
    \node[observed]  (z-med) at (0,0)   {$Z$};
    \node[treatment] (t-med) at (2.5,0) {$T$};
    \node[observed]  (m-med) at (5,0)   {$M$};
    \node[outcome]   (y-med) at (7.5,0) {$Y$};
    \node[latent]    (u-med) at (3.75,2) {$U$};
    \draw[arrow] (z-med) -- (t-med);
    \draw[arrow] (t-med) -- (m-med);
    \draw[arrow] (m-med) -- (y-med);
    \draw[arrow] (u-med) -- (t-med);
    \draw[arrow] (u-med) -- (m-med);
\end{scope}

\end{tikzpicture}}
    \caption{Causal graphs used in the simulation studies. Left: instrumental variable graph (IV). Centre: instrumental variable graph with a proxy variable (Proxy). Right: mediated instrumental variable graph (Mediated). $U$ is exogenous and the other nodes are endogenous categorical variables.}
    \label{fig:causal_graphs}
\end{figure*}

\section{Introduction}

Estimating the effects of interventions from observational data is a central problem in causal inference. In many applications, however, the available observations and causal assumptions do not uniquely determine the causal quantity of interest. Hidden confounding, incomplete measurements, or limited experimental access may permit multiple structural causal models to generate
the same observational distribution while implying different responses to an intervention. In such settings, reporting a single causal estimate can conceal uncertainty that persists even with complete knowledge of the observational
distribution.

Partial identification provides a principled alternative to point estimation
\cite{manski2003partial}. Rather than selecting a single observationally
compatible model, partial identification characterises the range of values
attained by a causal query over all models satisfying the assumed causal
structure and observational constraints. For example, the identified interval
for an average treatment effect is obtained by minimising and maximising the
effect over all structural models that reproduce the observed distribution.
The resulting interval separates uncertainty arising from causal
non-identifiability from ordinary uncertainty in parameter estimation.

For discrete causal models, causal bounds can be derived using linear,
multilinear, or polynomial optimisation. Classical results provide sharp
bounds in important special cases, such as instrumental-variable models
\cite{balke1997bounds}. More general approaches use logical relations between
counterfactual events, canonical structural causal models, and automated polynomial-programming formulations \cite{categorical_iv, sachs2023general,boushehrian26fundamental}. These approaches
offer valuable sharpness or certification guarantees, but the size of the
resulting optimisation problem can increase rapidly with the number of
variables, latent-variable configurations, and possible response functions. Furthermore, they have narrow applicability to general problems. 

Recently, generative neural network approaches have been applied to this problem \cite{BoundedCausalEffects, NeuralPartlID, balazadeh2022partial}. These attempt to find the extremal functions that satisfy the bound, providing both the bounds and functions that reach these bounds. While these are applicable to a wide range of problems, we find that their accuracy is limited. 

This paper introduces a differentiable approach to causal partial
identification for finite discrete structural causal models. Each structural
mechanism is represented by a neural network whose output defines a
conditional categorical distribution. The graph may contain explicit latent
variables, which are marginalised to obtain the induced distribution over the
observed variables. 

We fit a neural network model by minimising the Energy Distance \cite{energy_szekely} between the predictive model and the observational distribution, with a weighted causal loss.
The Energy distance is a measure of distributional disagreement, minimised if the predicted and true distributions are identical. A novel extrapolation procedure, with theoretical justification, is introduced to estimate the extremal causal effects compatible with the observational distribution. 

\section{Related Work}

\subsection{Optimisation-based causal bounds}
Partial identification characterises the range of causal quantities compatible with the observed distribution and the assumed causal structure \citep{manski2003partial}. For discrete structural causal models, classical approaches represent latent structural mechanisms through deterministic response functions and optimise over their distributions. This yields linear programmes in important settings, including the instrumental-variable model with binary \citep{balke1997bounds} and general categorical variables \citep{categorical_iv}. Symbolic methods extend this construction to broader classes of causal graphs and queries and are implemented in tools such as \texttt{causaloptim} \citep{sachs2023general,causaloptim_pkg}. More general discrete problems can be expressed as polynomial programs \cite{zhang_partial}. Autobounds automatically constructs such programs and applies spatial branch-and-bound to obtain progressively tighter certified bounds \citep{Autobounds}. Pruning and decomposition techniques can improve scalability for particular classes of graphs \citep{shridharan2023scalable}. These methods offer sharpness or certification guarantees in their supported settings, but their computational cost can grow rapidly with the number of variables and categories, and may fail to converge in more complex settings.

\subsection{Neural approaches to partial identification}
Neural causal models \cite{causal_neural_connection} provide an alternative to explicit response-function enumeration by representing structural mechanisms with flexible function approximators. Their connection to structural causal models and counterfactual reasoning has been studied theoretically, including representation and identifiability properties \citep{xia2022neural}. Bounded Counterfactual Causal Effects (BCCE) parametrises structural mechanisms and latent variables using implicit generative models \cite{gan} and employs adversarial training to match the observational distribution while optimising the causal effect \citep{BoundedCausalEffects}. Related approaches impose distributional constraints using Wasserstein or Sinkhorn discrepancies and optimise them through Lagrangian procedures \citep{balazadeh2022partial,NeuralPartlID}. NeuralPartialID further establishes approximation and consistency results for continuous problems under regularity conditions such as Lipschitz continuity \citep{NeuralPartlID}. Stochastic causal programming and implicit-generative approaches also address partial identification with continuous or multivariate treatments \citep{balazadeh2022partial, padh2023stochastic}. These methods broaden the range of causal models that can be considered. 

\subsection{Relation to our method}
Our method builds on the neural causal partial-identification approaches. However, there are two key differences. Firstly, rather than comparing generated and observed samples through a learnt discriminator or an optimal-transport subproblem, we directly compute the probability assigned to every observed state and evaluate a closed-form energy distance between the resulting probability vectors. This avoids adversarial training and stochastic sampling of categorical model outputs.

The main distinction is our treatment of the exact-fit limit. Existing neural approaches attempt to optimise the causal query subject to, or jointly with, an observational-distribution constraint. In practice, this optimisation becomes difficult when the distributional discrepancy is required to be extremely small. We instead solve a signed family of stable penalised problems and extrapolate the causal contrast to the zero-penalty limit. Our perturbation analysis shows that, within a locally regular regime, the causal contrast varies linearly with the penalty parameter while the energy-distance loss varies quadratically. This provides both a principled extrapolation rule and a diagnostic for selecting the range over which extrapolation is valid.

\section{Preliminaries} \label{sec:prelim}

We introduce the causal and probabilistic notation used throughout.
Random variables are denoted by upper case letters (e.g., $U,X$), their realisations by lower case letters (e.g., $u,x$), and vectors/sets of variables by bold letters (e.g., $\mathbf U, \mathbf X$).
For a vector $\mathbf{X} = (X_1,\ldots,X_n)$, where each $X_i$ takes values in the finite set $\mathcal X_i$, the joint state space is 
\begin{equation}\label{eq:statespace}
\mathcal S=\prod_{i=1}^{n}\mathcal{X}_i
    =
    \{\mathbf s_1,\ldots,\mathbf s_K\},
\end{equation}
where $K=\prod_{i=1}^{n}\vert\mathcal{X}_i\vert$.
A probability distribution on $\mathcal S$ can be represented by an element of the simplex 
\[
    \Delta^{K-1}
    =
    \left\{
        p\in\mathbb{R}^{K}_{\geq 0}
        :
        \sum_{k=1}^{K}p_k=1
    \right\}.
\]

\subsection{Structural equation models and causal graphs}

A \emph{structural equation model} (SEM) \cite{pearl_2009} is a tuple
\[
M = (\mathbf U, \mathbf X,\mathbf F, P_{\mathbf U}),
\]
where
\(\mathbf U = (U_1, \dots, U_m)\) is a set of independent \emph{exogenous} variables, 
\(P_{\mathbf U}\) is their probability distribution,
\(\mathbf X = (X_1, \dots, X_n)\) is a set of \emph{endogenous} variables, 
and
\(\mathbf F = (f_1, \dots, f_n)\) is a set of structural functions.
Each endogenous variable is generated according to
\[
    X_i = f_i\bigl(\pa{\mathbf X}(i), \pa{\mathbf U}(i)\bigr),
\]
where \(\pa{\mathbf X}(i) \subseteq \mathbf X \setminus \{X_i\}\)
and
\(\pa{\mathbf U}(i) \subseteq \mathbf U\)
denote the set of endogenous and exogenous parents of \(X_i\), respectively.

We assume that the directed relations between endogenous variables are acyclic.
Consequently, the structural assignments can be evaluated in topological order.
This induces a distribution on the endogenous variables which we denote by $P_M(\mathbf X)$.
The SEM can be represented by an acyclic directed graph (DAG) with vertices $\mathbf X\cup\mathbf U$, with an edge $X_i\to X_j$ whenever $X_i\in\pa{\mathbf X}(j)$, and an edge $U_i\to X_j$ whenever $U_i\in\pa{\mathbf U}(j)$.
We denote this graph by $\g$.

Throughout this work, we consider a setting where $\mathbf X$ consists of discrete variables and $\mathbf U$ follows an unknown distribution that may be continuous or discrete. 
Sampling \(\mathbf{U}\sim P_{\mathbf{U}}\) and evaluating
the structural equations in topological order induces an
observational distribution on
\(\mathcal{S}\). 

\subsection{Causal estimands and partial identification}\label{sec:partial-identification}

A \emph{do-intervention} on a variable $(X_i \in \mathbf X)$, setting it to a value $x$, is denoted by $
\operatorname{do}(X_i=x)$. It modifies the structural equation model by replacing the structural equation for $X_i$,
\[
X_i = f_i\left(\pa{X}(i), \pa{U}(i)\right),
\]
with the constant assignment
\[
X_i := x.
\]
All other structural equations remain unchanged.
The intervention removes the causal influence of the original parents.
This corresponds to removing all incoming edges to $X_i$ in the causal DAG representing the model.
We denote the resulting interventional distribution as 
\[
P_M\bigl(\mathbf X\setminus\{X\} \mid \operatorname{do}(X=x)\bigr).
\]

Let $a$ denote a scalar causal estimand, and $a(M)$ denote this estimand evaluated in model $M$.
For instance, $a(M)$ can denote the average treatment effect (ATE) of a binary treatment $X$ on an outcome $Y$:
\begin{equation*}
\begin{split}
    \operatorname{ATE}_{X\to Y}&(M):=
    \E_M[Y\mid do(X=1)]
    -
    \E_M[Y\mid do(X=0)],
\end{split}
\end{equation*}
which describes the expected causal difference on $Y$ between applying a positive intervention (1) or applying a negative intervention (0). 

Given a causal graph $\g$ and the observed distribution $P_{\mathrm{obs}}$ of the variables $\mathbf X$,
the observational equivalence class is
\begin{equation}\label{eq:equiv}
\begin{split}
    &\mathcal{M}(\g,P_{\mathrm{obs}})
    =
    \left\{
        M :
        M \text{ is compatible with } G
        \text{ and }
        P_M(\mathbf X)=P_{\mathrm{obs}}
    \right\},
\end{split}
\end{equation}
that is, the set of SEMs where parent sets agree with $\g$, and the induced distribution over $\mathbf X$ agrees with $P_{\mathrm{obs}}$.
The identified set of a
causal estimand $a$ given $\g$ and $P_{\mathrm{obs}}$ is then
\[
    \mathcal{I}_a(\g,P_{\mathrm{obs}})
    =
    \left\{
        a(M):
        M\in\mathcal{M}(\g,P_{\mathrm{obs}})
    \right\}.
\]
When this set is a singleton, we say $a$ is \emph{point-identified}.
Otherwise, it is only \emph{partially identified}, and its (sharp) lower and upper bounds are 
\begin{equation}\label{eq:bounds}
    \underline{a}
    =
    \inf_{M\in\mathcal{M}(\g,P_{\mathrm{obs}})} a(M),
    \quad
    \overline{a}
    =
    \sup_{M\in\mathcal{M}(\g,P_{\mathrm{obs}})} a(M).
\end{equation}
When the extrema are attained, the infimum and supremum may
be replaced by a minimum and maximum. These bounds describe
uncertainty due to causal non-identifiability even when
\(P_{\mathrm{obs}}\) is known exactly.


For the rest of this work, we focus on the ATE, although our method is applicable to arbitrary causal estimands.

\subsection{Energy distance}

Let \(P\) and \(Q\) be probability distributions on \(\mathbb{R}^d\). Let
\(\mathbf X,\mathbf X' \overset{\mathrm{iid}}{\sim} P\) and
\(\mathbf Y,\mathbf Y' \overset{\mathrm{iid}}{\sim} Q\), with all random variables mutually
independent. For \(0 < \beta < 2\), we use the following convention for the energy distance between
\(P\) and \(Q\) \cite{energy_distance}.
\begin{equation}
\operatorname{ED}_{\beta}(P,Q)
=
2\,\mathbb{E}\lVert \mathbf X-\mathbf Y\rVert_2^\beta
-
\mathbb{E}\lVert \mathbf X-\mathbf X'\rVert_2^\beta
-
\mathbb{E}\lVert \mathbf Y-\mathbf Y'\rVert_2^\beta.
\label{eq:energy_distance}
\end{equation}


The energy distance measures distributional discrepancy through expected pairwise distances and therefore does not require a probability density evaluation. Because it can be estimated using only samples and is differentiable, it can also serve as a loss function for training generative models without requiring access to the model likelihood.

In the finite discrete case considered in this paper, define the pairwise distance matrix
\begin{equation*}
    \mathbf D_{k\ell}
    =
    \|\mathbf s_k-\mathbf s_\ell\|_2^\beta,
    \qquad
    0<\beta<2.
\end{equation*}

\label{prop:matrix_energy_score}

\begin{proposition} [Finite-support energy distance] \label{theorem:energy_dist}
Let \(\mathbf{p},\mathbf{q}\in\Delta^{K-1}\) be the probability
vectors on \(\mathcal{S}\). Then
\begin{equation}
\label{eq:discrete_energy_distance}
    \begin{split}\mathrm{ED}_\beta(\mathbf p,\mathbf q)
    &=2\mathbf p^\top\mathbf D\mathbf q
    -\mathbf p^\top\mathbf D\mathbf p
    -\mathbf q^\top\mathbf D\mathbf q
    =
    -(\mathbf p-\mathbf q)^\top
    \mathbf D
    (\mathbf p-\mathbf q).
    \end{split}
\end{equation}
\end{proposition}
The full proof is given in Appendix~\ref{app:finite_support_energy_distance}. In short, since we have finite support, every pairwise distance can be represented as an entry in $\mathbf{D}$, which leads to the simplified weighted probability sum.

Since $\mathbf p-\mathbf q$ has zero sum and
$\mathbf D$ is strictly conditionally negative definite for distinct support
points and $0<\beta<2$ \cite{schoenberg_theorem}, this quantity is
nonnegative and vanishes only when $\mathbf p=\mathbf q$.

\section{Method}
We describe our method for estimating causal estimands.

\subsection{Neural parametrisation of discrete models}\label{sec:parameterisation}

Recall from \Cref{sec:prelim} that \(\mathbf{X}=(X_1,\ldots,X_n)\) denotes the endogenous
variables and that \(X_i\) takes values in the finite state
space \(\mathcal{X}_i\).
Define $k_i = \vert \mathcal X_i\vert,$
and index the elements of $\mathcal X_i$ by $j\in\{0,\ldots,k_i-1\}$.

For each endogenous variable $X_i$, we introduce a neural network (NN) with parameters $\theta_i$.
Given $\mathbf x_{\pa{\mathbf X}(i)}$ and $
    \mathbf u_{\pa{\mathbf U}(i)}$,
the values of endogenous and exogenous parents of $X_i$, respectively, the neural network produces a categorical probability vector 
\begin{equation} \label{eq:pi_theta_param}\begin{split}
    \boldsymbol\pi_{i}&(\mathbf x_{\pa{\mathbf X}(i)},
    \mathbf u_{\pa{\mathbf U}(i)};\theta_i) = 
    \operatorname{softmax}\left(g_i(\mathbf x_{\pa{\mathbf X}(i)},
    \mathbf u_{\pa{\mathbf U}(i)}; \theta_i)\right) \in \Delta^{k-1},
\end{split}
\end{equation}
where $g_i(\cdot;\theta_i)$
is the NN corresponding to $X_i$ parametrised with $\theta_i$, and
a Softmax activation is applied to the output layer to ensure that the vector $\boldsymbol{\pi}$ lies in the probability simplex.
We write
\[
    \pi_{i,j}
    \left(
        \mathbf{x}_{\pa{\mathbf{X}}(i)},
        \mathbf{u}_{\pa{\mathbf{U}}(i)};
        \theta_i
    \right)
\]
for the probability of $X_i$ taking value $j$.
A separate network is used for each endogenous variable, and
we collect all parameters as
\(\theta=(\theta_1,\ldots,\theta_n)\).


Latent variables are modelled as i.i.d uniform random variables, 
\begin{align*}
    U_i \sim U(\mathbf{0}, \mathbf{1}). 
\end{align*}
This does not restrict our method to Uniform latent variables, since deep NNs can arbitrarily transform the latents to approximate the underlying latents \cite{nn_approximation}. 

The parametrised mechanisms together with the exogenous distribution, define a neural SEM \(M_\theta\).
Under this parametrisation, 
the observational and interventional distributions under $M_\theta$ can be computed as follows.

\paragraph{Observational distribution.}
For a subset $\mathbf V\subseteq \mathbf X$ of endogenous variables, let
\(\mathbf s_{k,\mathbf V}\) denote the values of \(\mathbf V\) in state \(\mathbf s_k\); see \Cref{eq:statespace}.
Conditional on a realisation \(\mathbf{u}\) of the exogenous
variables, the probability of \(\mathbf s_k\) under the model $M_\theta$ factorises according to
\begin{equation}
    \label{eq:conditional-state-probability}
    p_\theta(\mathbf s_k\mid\mathbf u)
    =
    \prod_{i=1}^{n}
    \pi_{i,\mathbf s_{k_i}}
    \left(
        \mathbf{s}_{k,\operatorname{Pa}_{\mathbf{X}}(i)},
        \mathbf{u}_{\operatorname{Pa}_{\mathbf{U}}(i)};
        \theta_i
    \right).
\end{equation}
Define the corresponding conditional probability vector
\begin{equation*}
    \mathbf{p}_\theta(\mathbf{u})
    =
    \left(
        p_\theta(\mathbf s_1\mid\mathbf{u}),
        \ldots,
        p_\theta(\mathbf s_K\mid\mathbf{u})
    \right)^\top
    \in\Delta^{K-1}.
\end{equation*}
Marginalising the exogenous variables gives the observational
distribution induced by \(M_\theta\):
\begin{equation}
    \label{eq:marginal-model-distribution}
    \mathbf{p}_\theta
    =
    \mathbb{E}_{\mathbf{U}}
    \left[
        \mathbf{p}_\theta(\mathbf{U})
    \right]
    =
     \left(
        p_\theta(\mathbf s_1),
        \ldots,
        p_\theta(\mathbf s_K)
    \right)^\top
    \in\Delta^{K-1}.
\end{equation}




This formulation gives a differentiable way to train our discrete model since \Cref{eq:conditional-state-probability} is a product of predicted probabilities and
for each exogenous realisation $\mathbf u$,
the complete vector
\(\mathbf{p}_\theta(\mathbf{u})\) is computed by enumerating
the finite endogenous state space.
We therefore do not sample
categorical endogenous variables and do not require a Gumbel-Softmax \cite{gumbel_softmax} approximation.
Furthermore, unlike policy-gradient, we do not need to stochastically sample over the categorical predictions.


The node networks also need not be evaluated separately for
every complete state.
Conditional independencies make computing the full vector $\mathbf p_\theta$ significantly faster. 
For each node \(X_i\), outputs can be cached
for each distinct configuration of its parents.
Consequently, the number of node-network evaluations per
exogenous draw is
\begin{equation*}
\sum_{i=1}^n\prod_{j\in\textrm{Pa}_\mathbf{X}(i)} \vert \mathcal X_j\vert,
\end{equation*}
where the product is one for a node without endogenous
parents.
For the ``straight line'' graph where nodes form a simple chain, the compute complexity is $\mathcal{O}(n k)$.
Further implementation details are provided in Appendix \ref{apx:implementation}

\paragraph{Interventional distribution.}
The same representation permits direct
calculation of interventional distributions. 
Consider an
intervention \(do(X_t=x)\). Conditional on
\(\mathbf{U}=\mathbf{u}\), a state \(\mathbf s_k\) is assigned zero
probability unless \(\mathbf s_{k_t}=x\), i.e., it is a compatible state. 
For compatible states, the
factor associated with the intervened variable is removed:
\begin{align*}
    &p_\theta
    \left(
        \mathbf s_k\mid do(X_t=x),\mathbf u
    \right)
     =
    \mathbf{1}\{\mathbf s_{k_t}=x\}
    \prod_{\substack{i=1\\i\neq t}}^{n}
    \pi_{i,\mathbf s_{k_i}}
    \left(
        \mathbf{s}_{k,\pa{\mathbf{X}}(i)},
        \mathbf{u}_{\pa{\mathbf{U}}(i)};
        \theta_i
    \right).
\end{align*}
This is the truncated factorization associated with an intervention, rather
than ordinary conditioning on \(X_t=x\).
Marginal interventional probabilities $p_\theta\left(\mathbf s_k\mid do(X_t=x)\right)$ are computed similarly to \Cref{eq:marginal-model-distribution}.
If \(Y=X_r\) is the outcome variable, then
\begin{equation*}
    \E_\theta
    \left[
        Y\mid do(X_t=x)
    \right]
    =
    \sum_{k=1}^{K}
    \mathbf s_{k_r}\,
    p_\theta
    \left(
        \mathbf s_k\mid do(X_t=x)
    \right),
\end{equation*}
where $\E_\theta$ denotes expectation under the model $M_\theta$ with slight abuse of notation.

\subsection{Empirical energy-distance loss}

Here, we explain how the Energy Distance (\Cref{eq:discrete_energy_distance}) is estimated as a loss function for our training procedure.
Let
\[
    \mathcal{D}_N
    =
    \left\{
        \mathbf{x}^{(1)},\ldots,\mathbf{x}^{(N)}
    \right\}
\]
be an observational dataset of \(N\) independent samples
from \(P_{\mathrm{obs}}\), with each
\(\mathbf{x}^{(\ell)}\in\mathcal{S}\). We represent its
empirical distribution by the probability vector
\[
    \widehat {\mathbf q}=(\hat q_1,\ldots,\hat q_K)^\top\in\Delta^{K-1},
\]
where
\begin{equation}
    \hat q_k
    =
    \frac{1}{N}
    \sum_{\ell=1}^{N}
    \mathbf{1}
    \left\{
        \mathbf{x}^{(\ell)}=\mathbf s_k
    \right\},
    \qquad k=1,\ldots,K.
    \label{eq:empirical-distribution}
\end{equation}
Thus, \(q\) is the empirical analogue of the observational
distribution \(P_{\mathrm{obs}}\) introduced in \Cref{sec:partial-identification}.
Further, for a neural network model $M_\theta$, the expectation in
\Cref{eq:marginal-model-distribution} is
approximated using \(B\) independent exogenous draws
\(\mathbf{u}^{(1)},\ldots,\mathbf{u}^{(B)}\):
\begin{equation*}
    \widehat{\mathbf{p}}_\theta
    =
    \frac{1}{B}
    \sum_{b=1}^{B}
    \mathbf{p}_\theta\big(\mathbf{u}^{(b)}\big).
\end{equation*}
Directly substituting $\widehat{\mathbf p_\theta}$
into the quadratic term $\mathbf p^\top\mathbf D\mathbf p$ in \Cref{eq:discrete_energy_distance} is biased because
the same Monte Carlo sample then appears in both factors \cite{energy_score},
introducing diagonal self-pair terms.
We instead use a differentiable
U-statistic.
An unbiased estimator of \(\mathbf p_\theta^\top\mathbf  D\mathbf p_\theta\) is
\begin{align*}
    \widehat{C}_\theta
    &=
    \frac{1}{B(B-1)}
    \sum_{\substack{b,b'=1\\b\neq b'}}^{B}
    \left(\mathbf p_\theta\big(\mathbf{u}^{(b)}\big)\right)^\top
    \mathbf D
    \mathbf p_\theta\big(\mathbf{u}^{(b')}\big) \\
    &=
    \frac{B}{B-1}
    \widehat{\mathbf p}_\theta^\top
    \mathbf D
    \widehat{\mathbf p}_\theta 
    -
    \frac{1}{B(B-1)}
    \sum_{b=1}^{B}
    \left(\mathbf p_\theta\big(\mathbf{u}^{(b)}\big)\right)^\top
    \mathbf D
    \mathbf p_\theta\big(\mathbf{u}^{(b)}\big).
\end{align*}
The empirical loss used for optimization is therefore
\begin{equation}
    \label{eq:empirical-energy-loss}
    \widehat{\mathcal L}_{\mathrm{ED}}(\theta;\beta):={\operatorname{ED}}_\beta(\widehat{\mathbf p}_\theta;\widehat{\mathbf q})
    =
    2\widehat{\mathbf p}_\theta^\top\mathbf D\widehat{\mathbf q}
    -
    \widehat{C}_\theta
    -
    \widehat{\mathbf q}^\top\mathbf D\widehat{\mathbf q}.
\end{equation}
Conditional on \(\widehat{\mathbf q}\), this is an unbiased estimator of
\(\widehat{\operatorname{ED}}_\beta(\theta;\widehat{\mathbf q})\). Unlike the
population Energy Distance, a finite-sample realization of
the U-statistic may be slightly negative. This does not
affect its differentiability or its use as an optimization
objective.
We use
\(\beta=1\) in the rest of the manuscript and suppress the $\beta$ parameter.

Every term in
Equation~\eqref{eq:empirical-energy-loss} is differentiable
with respect to \(\theta\), because the conditional
probability vectors \(\mathbf p_\theta (\mathbf u^{(b)})\) are computed directly
from products of neural-network outputs. Further details of
the estimator and its implementation are given in
Appendix~\ref{apx:implementation}.

\subsection{Penalised optimisation formulation}
Given the causal graph $\g$ and the empirical observational distribution
\(\widehat{\mathbf q}\) from \Cref{eq:empirical-distribution},
the observational equivalence class of \Cref{eq:equiv} within the neural network model is 
\begin{align*}
    \Theta(\g, \widehat {\mathbf q})
    &=
    \left\{
        \theta\in\Theta_\g:
        {\mathbf p}_\theta=\widehat {\mathbf q}
    \right\} =
    \left\{
        \theta\in\Theta_\g:
        {\mathcal{L}}_{\mathrm{ED}}
        (\theta;\widehat{\mathbf q})=0
    \right\}
\end{align*}
by \Cref{theorem:energy_dist}. 

For brevity, we write $a(\theta):=a(M_\theta)$.
The lower and upper causal bounds within the neural model
family are therefore
\begin{equation*}
    \underline a_{\Theta}
    =
    \min_{\theta\in\Theta(\g,\widehat{\mathbf q})}
    a(\theta),
    \qquad
    \overline a_{\Theta}
    =
    \max_{\theta\in\Theta(\g,\widehat{\mathbf q})}
    a(\theta).
\end{equation*}
These are neural-model, finite-sample analogues of the sharp
population bounds defined in \Cref{eq:bounds}.
%
%
%
%
%
Rather than solving these constrained optimisation problems directly, 
for a scalar penalty
parameter \(\lambda\in\mathbb{R}\), we consider the penalised optimisation problems
\begin{align} \label{eq:SoftOptim}
\theta_\lambda=
\arg\min_{\theta}J_\lambda(\theta),
\end{align}
where
\[J_\lambda(\theta):=
{\mathcal{L}}_{\mathrm{ED}}(\theta;\widehat{\mathbf{q}}) + \lambda a(\theta).\]
The sign of \(\lambda\) determines which extremum is
favoured. 
For \(\lambda>0\), minimising
Equation~\eqref{eq:SoftOptim} favours smaller
values of \(a(M_\theta)\) and therefore traces the lower-bound. For \(\lambda<0\), it favours larger values and
traces the upper-bound.
As \(|\lambda|\) decreases, observational fit receives
increasing priority relative to the causal estimand. In the
ideal case, the one-sided limits
\(
    \lambda\downarrow0\)
and
\(
    \lambda\uparrow0
\)
approach models attaining
\(\underline a_{\Theta}\) and
\(\overline a_{\Theta}\), respectively.

In practice, however, optimisation becomes unreliable when
\(|\lambda|\) is extremely small
(see \Cref{fig:full_lambda}). 
We therefore
do not take the estimate from the smallest available value of
\(|\lambda|\). Instead, we train at a series of stable nonzero
values of $\lambda$ and extrapolate the resulting causal estimates to
\(\lambda=0\), as explained below.

\subsection{Extrapolation}
Let \(
    I_+=(0,\varepsilon)\), and \(
    I_-=(-\varepsilon,0)
\)
for some \(\varepsilon>0\).
\begin{assumption}[smooth one-sided solution paths]\label{as:smooth}
For each \(s\in\{+,-\}\), there exists a map
\[
    \theta_s:I_s\rightarrow\Theta_G
\]
such that \(\theta_s(\lambda)\) is a local minimiser of
\(J_\lambda\) for every \(\lambda\in I_s\). Moreover, the
induced observational distribution
\(
    \mathbf p_{\theta_s(\lambda)}
\)
and causal estimand
\(
    a(\theta_s(\lambda))
\)
admit twice continuously differentiable one-sided
extensions to \(\lambda=0\), with
\(
     \mathbf p_{\theta_s(0)}=\widehat{\mathbf q}\).
\end{assumption}
\Cref{as:smooth} concerns the distributions and causal estimands
represented by the neural networks; it does not require the
raw parameter vectors \(\theta_s(\lambda)\) to converge.
Appendix \ref{app:kkt_perturbation} provides sufficient local KKT conditions under which \Cref{as:smooth} holds.

\begin{theorem}[Local extrapolation]
\label{thm:local-extrapolation}
Under \Cref{as:smooth}, for each \(s\in\{+,-\}\), there exist
constants \(c_s\in\mathbb{R}\) and \(d_s\geq0\) such that,
as \(\lambda\to0\) within \(I_s\),
\begin{align}
    &a(\theta_s(\lambda))
    =
    a(\theta_s(0))
    +
    c_s\lambda
    +
    O(\lambda^2),
    \label{eq:local-estimand-expansion}\\
    &\mathcal L_{\mathrm{ED}}
    \left(
        \theta_s(\lambda);
        \widehat{\mathbf q}
    \right)
    =
    d_s\lambda^2
    +
    O(|\lambda|^3).
\end{align}
\end{theorem}

\begin{proof}
Since \(a(\theta_s(\lambda))\) is twice continuously differentiable at zero,
its one-sided Taylor expansion gives
\[
    a(\theta_s(\lambda))
    =
    a(\theta_s(0))
    +
    a'(\theta_s(0))\lambda
    +
    O(\lambda^2),
\]
where $a'(\theta_s(0)):=\frac{d}{d\lambda}a(\theta_s(\lambda))\big|_{\lambda=0}$.
This proves
Equation~\eqref{eq:local-estimand-expansion} with
\(c_s=a'(\theta_s(0))\).

Similarly, let
\[
    \mathbf v_s
    =
    \frac{\mathrm{d}}{\mathrm{d}\lambda}\mathbf p_{\theta_s(\lambda)}\big|_{\lambda=0}.
\]
Using twice differentiability and \Cref{as:smooth}, 
\[
    \mathbf p_{\theta_s(\lambda)}
    -
    \widehat{\mathbf q}
    =
    \mathbf p_{\theta_s(\lambda)}
    -
    \mathbf p_{\theta_s(0)}
    =
    \lambda\mathbf v_s
    +
    O(\lambda^2).
\]
Substitution into the expression in
\Cref{theorem:energy_dist} gives
\begin{align*}
    \mathcal L_{\mathrm{ED}}
    \left(
        \theta_s(\lambda);
        \widehat{\mathbf q}
    \right)
    =
    -
    \left(
        \mathbf p_{\theta_s(\lambda)}
        -
        \widehat{\mathbf q}
    \right)^\top
    \mathbf D
    \left(
        \mathbf p_{\theta_s(\lambda)}
        -
        \widehat{\mathbf q}
    \right)=
    -
    \lambda^2
    \mathbf v_s^\top \mathbf D\mathbf v_s
    +
    O(|\lambda|^3).
\end{align*}
Because \(\mathbf p_{\theta_s(\lambda)}\) is a probability vector,
\(\mathbf 1^\top\mathbf v_s=0\). Conditional negative
definiteness of \(\mathbf D\) implies $\mathcal L_{\mathrm{ED}}$ is
nonnegative.
\end{proof}
A more detailed analysis is given in Appendix \ref{app:kkt_perturbation}. 
\begin{assumption}\label{as:global}
    For every sufficiently small \(\lambda\in I_s\),
\(\theta_s(\lambda)\) is a global minimiser of
\(
    J_\lambda(\theta)
    =
    \mathcal L_{\mathrm{ED}}(\theta;\widehat{\mathbf q})
    +
    \lambda a(\theta).
\)
\end{assumption}
\begin{corollary}\label{cor}
    Under Assumptions \ref{as:smooth} and \ref{as:global},
    \[
    a(\theta_{+}(0))
    =
    \underline a_\Theta(\widehat q),
    \qquad
    a(\theta_{-}(0))
    =
    \overline a_\Theta.
    \]
    Consequently, the intercepts in \Cref{eq:local-estimand-expansion} for $s=+$ and $s=-$
    equal the lower and
    upper neural causal bounds, respectively.
\end{corollary}

\begin{proof}
Let \(\theta^\star\) be any parameter such that $\mathbf p_{\theta^\star}=\widehat{\mathbf q}$, so that
\(\mathcal L_{\mathrm{ED}}(\theta^\star;\widehat{\mathbf q} )=0\). 
For
\(\lambda>0\), global optimality gives
\[
    \mathcal L_{\mathrm{ED}}(\theta_+(\lambda);\widehat q)
    +
    \lambda a(\theta_+(\lambda))
    \leq
    \lambda a(\theta^\star).
\]
Since the energy-distance loss is nonnegative,
\[
    a(\theta_+(\lambda))\leq a(\theta^\star).
\]
Taking \(\lambda\downarrow0\) and using
\(\mathbf p_{\theta_+(0)}=\widehat{\mathbf q}\) shows that \(a(\theta_+(0))\) is no larger
than the estimand of any exact-fit model $\theta^\star$. 
Hence
\(a(\theta_+(0))=\underline a_\Theta\). The argument for
\(\lambda<0\) is analogous, with the inequality reversed.
\end{proof}
Note that the branch derivatives $\mathbf v_s$ may differ for positive and
negative $\lambda$. Finally, the Energy Distance is well defined throughout
the simplex, unlike other (semi-)distances such as the Kullback--Leibler
divergence, which becomes unbounded at the boundary and exhibits steep
gradients and Hessians for small probabilities, potentially leading to more
poorly conditioned fits.

\subsection{Training and bound estimation}
Motivated by \Cref{thm:local-extrapolation} and \Cref{cor}, we train models at
several positive and negative values of \(\lambda\) and
extrapolate their causal estimates to \(\lambda=0\).
In
practice, we minimise the objective
\begin{equation} \label{eq:training_loss}
    \widehat J_\lambda(\theta)
    =
    \widehat {\mathcal{L}}_{\mathrm{ED}}
    (\theta;\widehat{\mathbf q})
    +
    \lambda a(\theta).
\end{equation}
A series of different models with different positive and negative $\lambda$ values are trained. 
We then regress the causal estimates on the $\lambda$
values on the positive and negative branches.
The intercept of the positive branch regression estimates the lower bound, while the one of the
negative branch estimates the upper bound. \emph{We do not solve the KKT conditions, rather we solve \Cref{eq:training_loss} using SGD}. 

The regression should use penalties within the local regime
described by \Cref{thm:local-extrapolation}. However, this analysis does not consider if the optima can be found, only how the local region around them behaves. Empirically, using SGD on NNs we observe that convergence near $\lambda=0$ is poor, but good convergence for small $\lambda$. Hence, we have a range of valid $\lambda$ between the convergence and locally linear regime. 

We select the regression points using the studentised-residual
procedure given in \cref{alg:linear_ate_bound} in the Appendix.
The method begins with the largest penalty magnitudes and
adds points closer to zero while they remain consistent with
the fitted linear trend.
The accompanying energy-distance values provide an
additional diagnostic. 

For all our experiments, we used 4 equally spaced points for positive and negative $\lambda$, $|\lambda| \in \{0.005,\allowbreak 0.010,\allowbreak 0.015,\allowbreak 0.020\}$. We found this range to be a good trade off between accuracy and training time, and is well within the linear regime. Further training details are given in Appendix \ref{apx:implementation}.



\section{Simulation Studies}
\begin{figure*}[t]
\centering
    \includegraphics[width=0.5\textwidth]{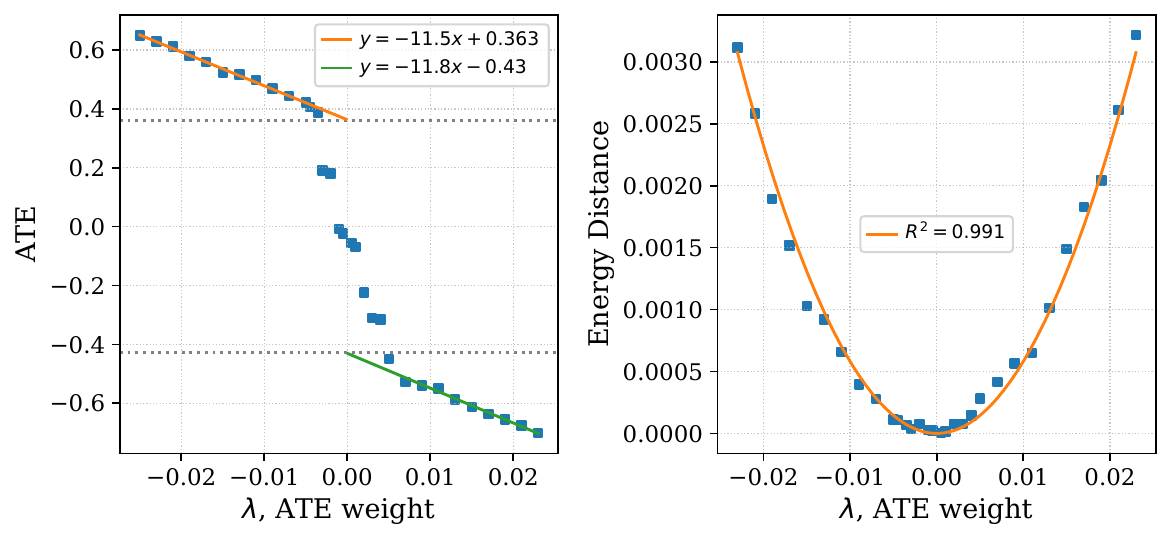}
    \includegraphics[width=0.5\textwidth]{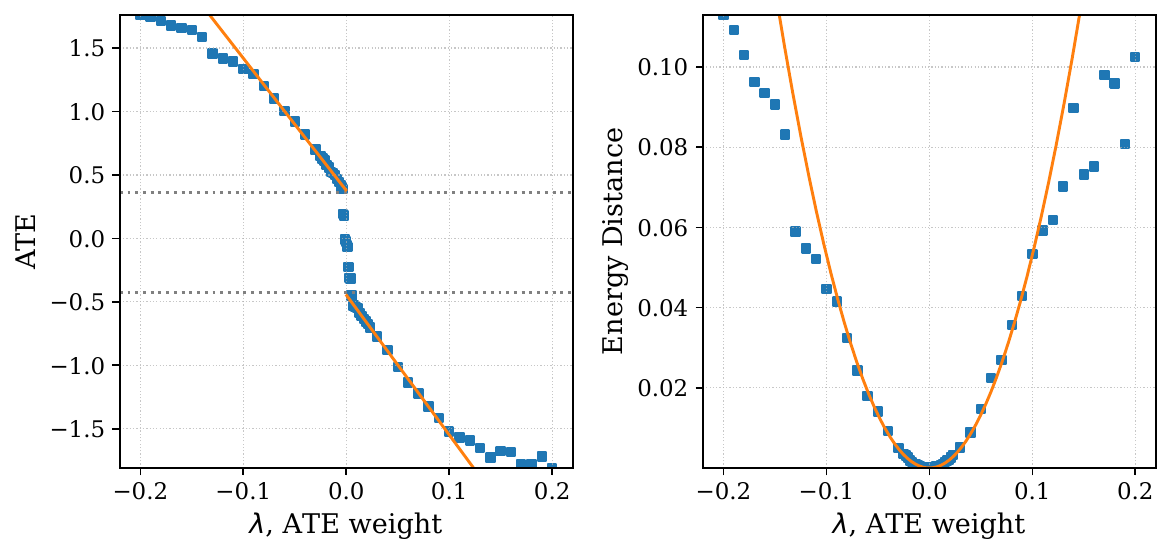}
    \caption{Model-predicted ATE and Energy Distance as $\lambda$ varies. Left: values in the extrapolation regime, where fitted ATE intercepts recover the predicted bounds and Energy Distance follows a quadratic trend. The linear fit and $R^2$ goodness-of-fit values are shown for the plots. Right: values across a wider range of $\lambda$ beyond the linear and quadratic regimes.}
    \label{fig:vary_lambda}
    \label{fig:full_lambda}
\end{figure*}

\newcommand{\result}[2]{%
  \(#1\text{\scriptsize\(\,\pm\,#2\)}\)%
}

\begin{table*}
\centering
\small
\setlength{\tabcolsep}{3pt}
\resizebox{\linewidth}{!}{%
\begin{tabular}{lcccccccc}
\hline
Dataset & \multicolumn{2}{c}{Autobounds} & \multicolumn{2}{c}{BCCE} & \multicolumn{2}{c}{NeuralPartialID} & \multicolumn{2}{c}{Ours} \\
\cline{2-9}
            & Upper & Lower & Upper & Lower & Upper & Lower & Upper & Lower \\
\hline
IV-Binary   & 0 & 0  &\result{0.135}{0.012} & \result{0.087}{0.021} & \result{0.105}{0.026} & \result{0.159}{0.041} & \result{0.020}{0.005} & \result{0.011}{0.004} \\
IV-Ternary  & 0 & 0 & \result{0.154}{0.043} & \result{0.125}{0.018} & \result{0.191}{0.065} & \result{0.162}{0.047} & \result{0.018}{0.011} & \result{0.010}{0.002} \\
Proxy       & \result{0.033}{0.011}* & \result{0.020}{0.006}* & \result{0.093}{0.027} & \result{0.107}{0.029} & \result{0.067}{0.020} & \result{0.077}{0.032} & \result{0.032}{0.012} & \result{0.026}{0.009} \\
Mediated    & \result{0.136}{0.038}* & \result{0.171}{0.041}* & \result{0.056}{0.020} & \result{0.037}{0.016} & \result{0.128}{0.053} &\result{0.099}{0.036} & \result{0.020}{0.009} & \result{0.021}{0.010} \\

\hline
\end{tabular}
}
\caption{Absolute error $\downarrow$ of the predicted upper and lower bounds for the ATE, compared with the true bounds. Standard errors are recorded over 10 datasets. (*) Autobounds does not fully converge and gives looser bounds than the true bound.}
\label{tab:main_results}
\end{table*}

\begin{figure}[t]
\centering
    \includegraphics[width=0.5\linewidth]{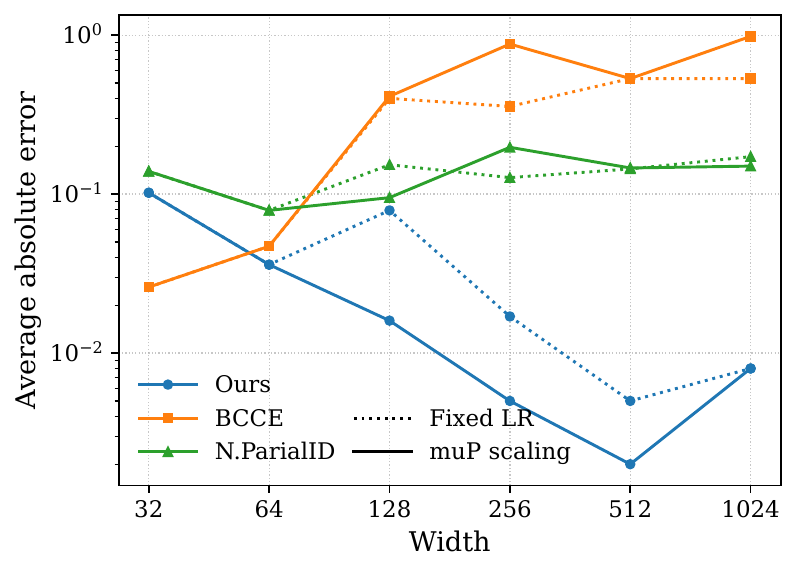}
    \caption{Absolute error between the analytical and predicted ATE bounds for our model, BCCE, and NeuralPartialID, averaged over the upper and lower bounds, over different neural network widths. Solid lines indicate a constant learning rate, while dashed lines indicate $\mu$P scaling.}
    \label{fig:width_ablation}
\end{figure}

We evaluate our method against two neural baselines: NeuralPartialID \cite{NeuralPartlID} and Bounding Confounded Causal Effects (BCCE) \cite{BoundedCausalEffects}, as well as Autobounds \cite{Autobounds}. 

Synthetic data is generated with four tasks derived from three causal graph structures: binary Instrumental Variable (IV), ternary IV, binary IV with a proxy variable (Proxy), and binary IV with a mediator (Mediated), shown in Figure \ref{fig:causal_graphs}.
For each task,
observational datasets are generated by randomly sampling from discrete SEMs whose mechanisms are defined by random Fourier functions.
Details are given in Appendix \ref{app:data_generation}. For the IV and Proxy graphs, the ground truths are derived from the CausalOptim package \cite{causaloptim_pkg}. However, CausalOptim and Autobounds are not able to handle the Mediated graph. For this type of graph, we derive the ground truth bounds in Appendix \ref{app:mediated_bounds}.

This procedure gives a range of datasets that are guaranteed to satisfy the graphical constraints, together with corresponding ground-truth ATE bounds, without requiring explicit identification of the extremal functions.  Each model was evaluated on 10 datasets for each graph type.
Performance is measured by the absolute error in the estimated lower and upper ATE
bounds, averaged across datasets.

\subsection{Bounding accuracy}

\Cref{tab:main_results} gives the average absolute error in the ATE upper and lower bounds. 
Autobounds gives exact bounds for the IV Binary and IV Ternary problems, but does not fully converge for the Proxy graph, where the bounds are looser than the analytical bounds. Our method consistently gives much more accurate bounds than BCCE and NeuralPartialID on all graph types. On the Proxy graph, our method is comparable to Autobounds, while on the Mediated graph, our method significantly outperforms Autobounds.

\subsection{Runtime}
We measure the wall-clock runtime of the neural methods on a fixed ternary IV dataset, with identical hardware (Appendix~\ref{app:hardware}). Our method takes 330 seconds, while BCCE takes 84 seconds and NeuralPartialID takes 399 seconds. 
Our method is $20\%$ faster than NeuralPartialID, though slower than BCCE.
These runtimes should be interpreted together with the
accuracy results, since the methods optimise different objectives and exhibit substantially different errors.
\subsection{Behaviour of our approach}
In order to visualise our method, on a fixed ternary IV dataset, we plot the ATE and $\mathcal{L}_\textrm{ED}$ for multiple values of $\lambda$ in \Cref{fig:vary_lambda}. The shape of both the ATE and Energy Distance curves agrees with our theory. A linear extrapolation for the ATE gives the correct bounds, and the Energy Distance is well modelled by a quadratic curve. However, the model does not converge to the extremal bounds near $\lambda=0$. Despite extensive tuning, this gap remains and necessitates our extrapolation method. 
The true ATE interval for this example is
\(
    [-0.427,\,0.361].
\)
Linear extrapolation gives
\(
    [-0.430,\,0.363],
\)
which closely matches the analytical bounds. In contrast,
using the trained models directly without extrapolation gives
approximately
\(
    [-0.450,\,0.388].
\)
Thus, the main improvement comes from estimating the
zero-penalty intercept rather than taking the causal
estimand from the model trained at the smallest available
penalty.
\Cref{fig:full_lambda} plots the ATE and $\mathcal{L}_\textrm{ED}$ for a wider range of $\lambda$. The ATE is linear in $\lambda$ for around $0.005<|\lambda|<0.09$. At $\lambda=0.09$, there is an abrupt change where both the ATE and energy distance stop following the linear and quadratic trend. Both changes are consistent with our theory; this corresponds to when a nodes predicted probability equals 0 or 1. There is a wide range of $\lambda$ within the linear regime, with our selected $\lambda$ inside it.

\subsection{Analysis of the baselines}  \label{section:baseline}
To investigate why the neural baselines underperform, we test each model using neural networks of varying widths. All other parameters are held constant, except for the learning rate, which is either scaled using $\mu$P scaling \cite{mup_scaling} or kept fixed. $\mu$P scaling promotes stable training across different widths by ensuring that the magnitudes of parameter updates and gradients remain stable. All models are trained on a fixed ternary IV SEM.  \Cref{fig:width_ablation} shows the average absolute error between the predicted and analytical bounds.

Ideally, improving the width and hence representational capacity of the model improves performance. However, this is not observed for NeuralPartialID and BCCE, the accuracy decreases for larger NN widths. For BCCE, we attribute this to an instability in the GAN training procedure where the critic becomes too strong, preventing the generator from learning. Similarly, for NeuralPartialID, this is attributed to unstable parameter updates and hyperparameter sensitivity. Further analysis can be seen in Appendix \ref{app:baseline_eval}. 

\section{Conclusion}

We introduced a neural approach to partial identification in finite discrete structural causal models. 
The method represents each
structural mechanism by a neural categorical model and
computes the induced observational and interventional
distributions explicitly over the finite state space. Causal
bounds are estimated by solving a signed family of
energy-distance-penalised problems and extrapolating the
resulting causal estimates to the zero-penalty limit.
Our analysis shows that, along a smooth one-sided solution
path, the causal estimand varies linearly with the penalty
parameter to first order, while the energy-distance loss
varies quadratically.
The
corresponding intercepts recover the lower and upper bounds
within the neural model family.
Across several simulations, our resulting extrapolated bounds are consistently more accurate than those of the evaluated neural baselines and remain competitive with analytical methods.

The approach is limited to discrete models and requires identifying a locally regular penalty range. Extending the approach to continuous variables and further development on the theory behind the validity of the linear extrapolation regime would broaden its practical scope. More generally, these results suggest that optimisation paths, rather than only their limiting iterates, can provide useful information for estimating causal bounds in partially identified models.

\clearpage
\nocite{Autobounds,causaloptim_pkg,NeuralPartlID,BoundedCausalEffects,RMSProp,schoenberg_theorem,balke1997bounds,gan_instablitiy,balazadeh2022partial,sinkhorn_loss,energy_distance,relu,adam,pytorch}
\bibliographystyle{plainnat}
\bibliography{aaai2027}

@inproceedings{NeuralPartlID,
  author  = {J. Tan and J. Blanchet and V. Syrgkanis},
  year    = {2024},
  title   = "{Consistency of Neural Causal Partial Identification}",
  booktitle = "{Advances in Neural Processing Systems 2024}",
}

@inproceedings{BoundedCausalEffects,
  title={A Generative Adversarial Framework for Bounding Confounded Causal Effects},
  author={Yaowei Hu and Yongkai Wu and Lu Zhang and Xintao Wu},
  booktitle={AAAI Conference on Artificial Intelligence},
  year={2021},
}

@article{Autobounds,
author = {Guilherme Duarte and Noam Finkelstein and Dean Knox and Jonathan Mummolo and Ilya Shpitser},
title = {An Automated Approach to Causal Inference in Discrete Settings},
journal = {Journal of the American Statistical Association},
volume = {119},
number = {547},
pages = {1778--1793},
year = {2024},
}

@Manual{causaloptim_pkg,
  title = {causaloptim: An Interface to Specify Causal Graphs and Compute Bounds on Causal Effects},
  author = {Michael C Sachs and Erin E Gabriel and Arvid Sjölander and Gustav Jonzon},
  year = {2026},
  url = {https://sachsmc.github.io/causaloptim/},
}

@article{sachs2023general,
  title={A general method for deriving tight symbolic bounds on causal effects},
  author={Sachs, Michael C and Jonzon, Gustav and Sj{\"o}lander, Arvid and Gabriel, Erin E},
  journal={Journal of Computational and Graphical Statistics},
  volume={32},
  number={2},
  pages={567--576},
  year={2023},
  publisher={Taylor \& Francis}
}

@inproceedings{mup_scaling,
  title={Tensor Programs V: Tuning Large Neural Networks via Zero-Shot Hyperparameter Transfer}, 
  author={Greg Yang and Edward J. Hu and Igor Babuschkin and Szymon Sidor and Xiaodong Liu and David Farhi and Nick Ryder and Jakub Pachocki and Weizhu Chen and Jianfeng Gao},
  year={2022},
  booktitle = "{Advances in Neural Processing Systems 2021}",
}

@inproceedings{gan_instablitiy,
author = {Salimans, Tim and Goodfellow, Ian and Zaremba, Wojciech and Cheung, Vicki and Radford, Alec and Chen, Xi},
title = {Improved techniques for training GANs},
year = {2016},
booktitle = {Proceedings of the 30th International Conference on Neural Information Processing Systems},
}

@book{pearl_2009,
author = {Judea Pearl},
year = {2009},
title = {Causality},
publisher = {Cambridge University Press}
}

@book{manski2003partial,
  title={Partial identification of probability distributions},
  author={Manski, Charles F},
  year={2003},
  publisher={Springer}
}

@article{balazadeh2022partial,
  title={Partial identification of treatment effects with implicit generative models},
  author={Balazadeh Meresht, Vahid and Syrgkanis, Vasilis and Krishnan, Rahul G},
  journal={Advances in Neural Information Processing Systems},
  volume={35},
  pages={22816--22829},
  year={2022}
}

@article{balke1997bounds,
  title={Bounds on treatment effects from studies with imperfect compliance},
  author={Balke, Alexander and Pearl, Judea},
  journal={Journal of the American statistical Association},
  volume={92},
  number={439},
  pages={1171--1176},
  year={1997},
  publisher={Taylor \& Francis}
}

@article{shridharan2023scalable,
  title={Scalable computation of causal bounds},
  author={Shridharan, Madhumitha and Iyengar, Garud},
  journal={Journal of Machine Learning Research},
  volume={24},
  number={237},
  pages={1--35},
  year={2023}
}

@article{xia2022neural,
  title={Neural causal models for counterfactual identification and estimation},
  author={Xia, Kevin and Pan, Yushu and Bareinboim, Elias},
  journal={arXiv preprint arXiv:2210.00035},
  year={2022}
}

@inproceedings{padh2023stochastic,
  title={Stochastic causal programming for bounding treatment effects},
  author={Padh, Kirtan and Zeitler, Jakob and Watson, David and Kusner, Matt and Silva, Ricardo and Kilbertus, Niki},
  booktitle={Conference on Causal Learning and Reasoning},
  pages={142--176},
  year={2023},
  organization={PMLR}
}

@inproceedings{gumbel_softmax,
title={Categorical Reparameterization with Gumbel-Softmax},
author={Eric Jang and Shixiang Gu and Ben Poole},
booktitle={International Conference on Learning Representations},
year={2017},
}

@article{energy_score,
author = {Tilmann Gneiting and Adrian E Raftery},
title = {Strictly Proper Scoring Rules, Prediction, and Estimation},
journal = {Journal of the American Statistical Association},
year = {2007},
}

@misc{nn_approximation,
      title={Optimal approximation of continuous functions by very deep ReLU networks}, 
      author={Dmitry Yarotsky},
      year={2018},
      archivePrefix={arXiv},
}

@inproceedings{sinkhorn_loss,
author = {Cuturi, Marco},
title = {Sinkhorn distances: lightspeed computation of optimal transport},
year = {2013},
booktitle = {Proceedings of the 27th International Conference on Neural Information Processing Systems - Volume 2},
series = {NIPS'13}
}

@misc{schoenberg_theorem,
    title={METRIC SPACES AND POSITIVE DEFINITE FUNCTIONS},
    author={1. J. Schoenberg},
    year={1938}
    }

@misc{relu,
author = {Nair, Vinod and Hinton, Geoffrey E.},
title = {Rectified linear units improve restricted boltzmann machines},
year = {2010},
}

@misc{adam,
      title={Adam: A Method for Stochastic Optimization}, 
      author={Diederik P. Kingma and Jimmy Ba},
      booktitle={International Conference on Learning Representations},
      year={2015},
}

@misc{RMSProp,
    title={RMSProp: Divide the gradient by a running average of its recent magnitude},
    author={Hinton},
    year={2012}
}

@article{energy_distance,
author = {Rizzo, Maria L. and Székely, Gábor J.},
title = {Energy distance},
journal = {WIREs Computational Statistics},
volume = {8},
number = {1},
pages = {27-38},
doi = {https://doi.org/10.1002/wics.1375},
year = {2016}
}

@inproceedings{causal_neural_connection,
author = {Xia, Kevin and Lee, Kai-Zhan and Bengio, Yoshua and Bareinboim, Elias},
title = {The causal-neural connection: expressiveness, learnability, and inference},
year = {2021},
booktitle = {Proceedings of the 35th International Conference on Neural Information Processing Systems},
}

@InProceedings{zhang_partial,
  title = 	 {Partial Counterfactual Identification from Observational and Experimental Data},
  author =       {Zhang, Junzhe and Tian, Jin and Bareinboim, Elias},
  booktitle = 	 {Proceedings of the 39th International Conference on Machine Learning},
  year = 	 {2022},
}

@article{energy_szekely,
author = {Gabor J. Szekely},
title = {E-statistics: The Energy of Statistical Samples},
journal = {Bowling Green State University,
Department of Mathematics and Statistics Technical Report},
year = {2003}
}

@misc{categorical_iv,
      title={The Categorical Instrumental Variable Model: Characterization, Partial Identification, and Statistical Inference}, 
      author={Yilin Song and F. Richard Guo and K. C. Gary Chan and Thomas S. Richardson},
      year={2025},
}

@misc{pytorch,
      title={PyTorch: An Imperative Style, High-Performance Deep Learning Library}, 
      author={Adam Paszke and Sam Gross and Francisco Massa and Adam Lerer and James Bradbury and Gregory Chanan and Trevor Killeen and Zeming Lin and Natalia Gimelshein and Luca Antiga and Alban Desmaison and Andreas Köpf and Edward Yang and Zach DeVito and Martin Raison and Alykhan Tejani and Sasank Chilamkurthy and Benoit Steiner and Lu Fang and Junjie Bai and Soumith Chintala},
      year={2019},
}

@misc{gan,
      title={Generative Adversarial Networks}, 
      author={Ian J. Goodfellow and Jean Pouget-Abadie and Mehdi Mirza and Bing Xu and David Warde-Farley and Sherjil Ozair and Aaron Courville and Yoshua Bengio},
      year={2014},
      booktitle={Proceedings of the International Conference on Neural Information Processing Systems}
}

@InProceedings{boushehrian26fundamental,
  title     = {Fundamental Limits and Optimal Methods for Sharp Analytical Causal Bounds in Instrumental Variable Models},
  author    = {Boushehrian, Arefe and Badri, Mohammad Reza and Akbari, Sina and Kiyavash, Negar},
  booktitle = {Proceedings of the 42nd Conference on Uncertainty in Artificial Intelligence},
  pages     = {689--744},
  year      = {2026},
  volume    = {337},
  series    = {Proceedings of Machine Learning Research},
  publisher = {PMLR}
}

\clearpage
\appendix
\begin{algorithm}[tb]
\caption{Lower (Upper) ATE bound estimation via studentized-residual fit}
\label{alg:linear_ate_bound}
\textbf{Input}: Data vectors $\mathbf{x}, \mathbf{y} \in \mathbb{R}^{L}$; minimum points $k_{\min}$; threshold $\tau$\\
\textbf{Output}: Intercept estimate $\hat{\beta}_0$
\begin{algorithmic}[1]
\STATE Sort $\{(\mathbf{x}_i, \mathbf{y}_i)\}_{i=1}^{L}$ by $\mathbf{x}_i$ descending (ascending), yielding $(\tilde{\mathbf{x}}, \tilde{\mathbf{y}})$.
\STATE $S \gets \{0, \dots, k_{\min}-1\}$ \hfill // initial set
\FOR{$i = k_{\min}, \dots, L-1$}
    \STATE $T \gets S \cup \{i\}$ \hfill // tentatively add next point
    \STATE Fit OLS: $\mathbf{y}_T \sim \beta_0 + \beta_1 \, \tilde{\mathbf{x}}_T$
    \STATE Compute internally studentised residuals $\mathbf{r}$.
    \IF{$|\mathbf{r}_{-1}| < \tau$}
        \STATE $S \gets T$ \hfill // accept
    \ELSE
        \STATE \textbf{break} \hfill // reject and stop
    \ENDIF
\ENDFOR
\STATE Fit $\tilde{\mathbf{y}}_{S} \sim \beta_0 + \beta_1 \, \tilde{\mathbf{x}}_{S}$ on final $S$.
\STATE \textbf{return} $\hat{\beta}_0$
\end{algorithmic}
\end{algorithm}
\section{Implementation Details} \label{apx:implementation}

Our model uses the assumed causal DAG to construct a neural structural causal
model. As in the data-generating process, the nodes are evaluated in
topological order. However, the unknown structural mechanisms are represented
by neural networks rather than fixed random functions.

Let $X_1,\ldots,X_m$ denote the discrete endogenous nodes and let
$\mathbf U$ denote the exogenous variables. For every node $X_i$, the model
specifies the categorical probabilities
\begin{equation*}
    \pi_{i,j}\bigl(\textrm{Pa}_X(i),\textrm{Pa}_U(i);\theta_i\bigr),
    \qquad j\in\{0,\ldots,k_i-1\},
\end{equation*}
where $\textrm{Pa}_X(i)$ and $\textrm{Pa}_U(i)$ are the endogenous and
exogenous parents of $X_i$, respectively. The exogenous variables are sampled
as in the main model, $\mathbf U\sim U(\mathbf{0},\mathbf{1})$.

For a node with $k_i$ possible states, the values of its parents are
concatenated into the input
\begin{equation*}
    \mathbf h_i^{(0)}
    =
    \operatorname{concat}\left(
        \textrm{Pa}_X(i),
        \textrm{Pa}_U(i)
    \right).
\end{equation*}
This input is passed through a multilayer perceptron,
\begin{align*}
    \mathbf h_i^{(\ell)}
    &=
    \rho\left(
        W_i^{(\ell)}\mathbf h_i^{(\ell-1)}
        +\mathbf b_i^{(\ell)}
    \right),
    \qquad \ell=1,\ldots,L,\\
    \mathbf a_i
    &=
    W_i^{(\mathrm{out})}\mathbf h_i^{(L)}
    +\mathbf b_i^{(\mathrm{out})},
\end{align*}
where \(\rho\) is the ReLU activation function and
\(\mathbf a_i\in\mathbb R^{k_i}\) contains the output logits. The conditional
probability of state $j$ is obtained using a softmax:
\begin{equation*}
    \pi_{i,j}\bigl(\textrm{Pa}_X(i),\textrm{Pa}_U(i);\theta_i\bigr)
    =
    \frac{\exp(a_{i,j})}
         {\sum_{\ell=0}^{k_i-1}\exp(a_{i,\ell})}.
\end{equation*}
A separate neural network with parameters $\theta_i$ is used for each node.

\subsection{Joint probability calculation}

The model computes probabilities explicitly over all possible joint
configurations of the discrete endogenous nodes. Let
\begin{equation*}
    \mathcal S
    =
    \prod_{i=1}^m
    \{0,\ldots,k_i-1\}
\end{equation*}
be the Cartesian product of their state spaces. The number of configurations
is therefore
\begin{equation*}
    K
    =
    \prod_{i=1}^m k_i.
\end{equation*}
We denote the $c$-th configuration by $\mathbf{s}_c\in\mathcal S$.

For a given realisation $\mathbf{u}$ of the exogenous variables, the
probability of a configuration factorizes according to the DAG:
\begin{equation}
    p_{\theta}\left(\mathbf{s}_c\mid\mathbf{u}\right)
    =
    \prod_{i=1}^m
    \pi_{i,s_{c,i}}\bigl(\textrm{Pa}_{\mathbf{s}_c}(i),
    \textrm{Pa}_{\mathbf{u}}(i);\theta_i\bigr).
    \label{eq:joint-neural-factorization}
\end{equation}
Thus, each network returns a probability vector over the states of its node,
and the entry corresponding to $s_{c,i}$ is selected for each global
configuration. The selected conditional probabilities are then multiplied
across nodes.

To avoid unnecessary network evaluations, the model identifies the unique
configurations of a node's discrete parents. The node network is evaluated
once for each unique parent configuration rather than once for every complete
configuration in $\mathcal S$. Its outputs are then mapped back to the
corresponding complete configurations. This optimization does not change the
probability represented by Equation~\eqref{eq:joint-neural-factorization}.

\subsection{Marginal probabilities}

The probabilities in Equation~\eqref{eq:joint-neural-factorization} are
conditional on a realisation of the exogenous variables. The model
distribution is obtained by integrating over $\mathbf U$:
\begin{equation*}
    p_{\theta}(\mathbf{s})
    =
    \mathbb E_{\mathbf U}
    \left[
        p_{\theta}\left(\mathbf{s}\mid\mathbf U\right)
    \right].
\end{equation*}
This expectation is approximated by Monte Carlo sampling. Given $B$
independent exogenous realisations,
\begin{equation*}
    \widehat p_{\theta}(\mathbf{s})
    =
    \frac{1}{B}
    \sum_{b=1}^{B}
    p_{\theta}\left(\mathbf{s}\mid\mathbf u^{(b)}\right).
\end{equation*}

\subsection{Interventional probabilities}

Interventions are implemented by replacing the model assignment for the
intervened node. Under the intervention
\(\operatorname{do}(X_t=a)\), configurations for which \(s_{c,t}\neq a\)
are assigned probability zero. For compatible configurations, the conditional
factor associated with \(X_t\) is omitted:
\begin{align*}
    p_{\theta}\left(
        \mathbf{s}_c\mid\operatorname{do}(X_t=a),\mathbf{u}
    \right) \notag\\
    =   
    \mathbf{1}\{s_{c,t}=a\}
    \prod_{i\in\{1,\ldots,m\}\setminus\{t\}}
    \pi_{i,s_{c,i}}\bigl(\textrm{Pa}_{\mathbf{s}_c}(i),
    \textrm{Pa}_{\mathbf{u}}(i);\theta_i\bigr).
\end{align*}
This is the truncated factorization associated with an intervention, rather
than ordinary conditioning on \(X_t=a\).

For an outcome node \(Y\), its expected value under treatment \(a\) is
computed directly from the probabilities of the enumerated states:
\
\begin{align*}
    \mathbb E_{\theta}
    \left[Y\mid\operatorname{do}(X_t=a)\right] \notag\\
    =
    \mathbb E_{\mathbf U}
    \left[
        \sum_{c=1}^{K}
        s_{c,Y}
        p_{\theta}\left(
            \mathbf{s}_c\mid\operatorname{do}(X_t=a),\mathbf U
        \right)
    \right].
\end{align*}
The causal contrast used in the main objective is consequently
\begin{align*}
    a(\theta)
    = 
    \mathbb E_{\theta}
    \left[Y\mid\operatorname{do}(X_t=1)\right]
    -
    \mathbb E_{\theta}
    \left[Y\mid\operatorname{do}(X_t=0)\right].
\end{align*}

\subsection{Estimating Energy Distance} \label{app:energy_dist}
\Cref{theorem:energy_dist} gives the Energy Distance for a random discrete model. However, when estimating the Energy Distance from samples from finite samples, a modified formula is used \cite{energy_distance}. Specifically, the term involving the distance between predictions, $\mathbb{E}||x-x'||$ is modified. The term can be estimated with the unbiased U estimator 
\begin{align*}
    \mathbb{E}||x-x'|| \approx \frac{1}{B^2} \sum_{i \neq j} \mathbf{p}^{i\top}_\theta \mathbf{D} \mathbf{p}^j_\theta \\
    = \frac{B}{B-1} \tilde{\mathbf{p}}_\theta \mathbf{D} \tilde{\mathbf{p}}_\theta - \frac{1}{B(B-1)}\sum_{i=1}^B \mathbf{p}^{i\top}_\theta \mathbf{D} \mathbf{p}&i_\theta
\end{align*}
where $\tilde{\mathbf{p}}$ is the mean predicted probability, or the biased V estimator
\begin{align*}
    \mathbb{E}||x-x'|| \approx \tilde{\mathbf{p}} \mathbf{D} \tilde{\mathbf{p}}.
\end{align*}
Using the V estimator means the Energy Distance is always positive, but biased positively. Alternatively, the U estimator gives an unbiased Energy Distance, but possibly negative empirical value. Since our method can handle negative distances, we use the U estimator. We note the difference between the two estimators is $\mathcal{O}(B^{-2})$ and makes very little difference for large batches. 

\subsection{Hyperparameters} \label{app:hardware}
We use deep neural networks for our model to parametrize  \Cref{eq:pi_theta_param}. The same model hyperparameters are used for all experiments. Our models have hidden width $1024$ and depth $3$. The input and output dimensions are determined based on the parent nodes and output categories. The ReLU activation function \cite{relu} is used for the hidden layers, with the Softmax activation on the output to constrain the output to valid probabilities. Our model is trained with the Adam optimiser \cite{adam} with a cosine decay schedule and maximum learning rate of 1e-4 and 24 epochs with 256 steps per epoch. The linear fit threshold was selected as $\tau = 1.25$. Our hyperparameters were tuned on a single Ternary IV graph, like the baselines. The random seed for each run is the same as the one used to generate the datasets, that is $\{0, \ldots,9\}$. 

Finally, our method is implemented using PyTorch \cite{pytorch} version 2.13. All experiments and timings (including baselines) are conducted using Ubuntu Linux 26.04. The hardware is as follows: GPU NVIDIA RTX 3090, CPU AMD RYZEN 5800X3D and 32GB DDR4 system memory. 

\section{Baselines and Evaluation Metrics}\label{apx:baseline}

For each dataset, we compute the observed distribution $\mathbf{q}$ using $B=1000000$ samples. The baselines compute their predicted bounds using these observed probabilities. 

The baselines are implemented using the official packages / code. 

Autobounds \cite{Autobounds} uses the default configuration, used from their Docker container. 

CausalOptim \cite{causaloptim_pkg} generates analytical bounds given a causal graph, given as a series of equations. For each graph type, we generate the analytical equations for the upper and lower bounds. 

NeuralPartialID \cite{NeuralPartlID} uses the official implementation. We tuned the hyperparameters using a sequential grid search to find the best parameters. The range of parameters searched were: Width [16, 32, 64, 128, 256, 512, 1024], depth [2, 3, 4], learning rate [3e-4, 1e-3, 3e-3], lagrange learning rate [0.1, 0.3, 0.5, 0.7]. Compared to the original implementation, we settled on increasing the depth of the NNs from 2 to 3, increasing the Lagrange learning rate from 0.1 to 0.5 and increasing the width from 32 to 64 with a max training length of 500 epochs, which we found gives the best performance on a standard task. 

BCCE \cite{BoundedCausalEffects} also uses the official implementation. We used the \textit{Nonlinear} model class for our experiments so the models are general non-linear neural networks. The default hyperparameters were very unstable and yielded poor results, so the parameters were manually tuned. For the final hyperparameters, the width of both the generator and critic was increased from 16 to 64, and the depth of the generator was increased to 3. The generator's optimiser was switched from RMSProp \cite{RMSProp} to AdamW, while RMSProp was kept for the discriminator. The learning rate was increased from 3e-4 to 5e-2 for both the discriminator and generator. Finally, a stepwise learning rate decay schedule was added of 0.5 every 200 epochs. These changes significantly increased the accuracy. Increasing the training duration from the default 2000 epochs did not improve performance, but led to training instabilities so the duration was fixed at 500 epochs. 

The hyperparameters for both BCCE and NeuralPartialID were configured on a Ternary IV graph sample. 

\section{Finite-Support Energy Distance}
\label{app:finite_support_energy_distance}

This section proves the finite-support representation in
\Cref{theorem:energy_dist}. Let
\(\mathbf{p}=(p_1,\ldots,p_K)^\top\) and
\(\mathbf{q}=(q_1,\ldots,q_K)^\top\). Let
\(\mathbf X,\mathbf X'\sim\mathbf p\) and
\(\mathbf Y,\mathbf Y'\sim\mathbf q\) independently.

\begin{proof}[Proof of \Cref{theorem:energy_dist}]
Expanding the cross-distribution expectation gives
\[
\begin{split}
\E\|\mathbf X-\mathbf Y\|_2^\beta
&=
\sum_{k,\ell=1}^{K}
P(\mathbf X=\mathbf s_k)P(\mathbf Y=\mathbf s_\ell)
\|\mathbf s_k-\mathbf s_\ell\|_2^\beta \\
&=\sum_{k,\ell=1}^{K} p_kq_\ell \mathbf D_{k\ell}
=\mathbf p^\top\mathbf D\mathbf q.
\end{split}
\]
Likewise,
\[
\E\|\mathbf X-\mathbf X'\|_2^\beta
=\sum_{k,\ell=1}^{K} p_kp_\ell \mathbf D_{k\ell}
=\mathbf p^\top\mathbf D\mathbf p,
\]
and
\[
\E\|\mathbf Y-\mathbf Y'\|_2^\beta
=\sum_{k,\ell=1}^{K} q_kq_\ell \mathbf D_{k\ell}
=\mathbf q^\top\mathbf D\mathbf q.
\]
Substituting these terms into \Cref{eq:energy_distance} gives
\[
\begin{split}
\mathrm{ED}_\beta(\mathbf p,\mathbf q)
&=2\mathbf p^\top\mathbf D\mathbf q
-\mathbf p^\top\mathbf D\mathbf p
-\mathbf q^\top\mathbf D\mathbf q \\
&=-(\mathbf p-\mathbf q)^\top\mathbf D(\mathbf p-\mathbf q),
\end{split}
\]
where the final equality uses the symmetry of \(\mathbf D\).
\end{proof}

By Schoenberg's theorem \cite{schoenberg_theorem}, \(\mathbf D\) is strictly
conditionally negative definite for distinct support points and
\(0<\beta<2\). Since \(\mathbf p-\mathbf q\) has zero sum, the quadratic
form above is nonnegative and is zero if and only if \(\mathbf p=\mathbf q\).

\section{KKT Perturbation Derivation}
\label{app:kkt_perturbation}

The main text uses softmax outputs, so the raw neural parameters $\theta$
are unconstrained while the conditional probability vectors are constrained
to the simplex implicitly. The KKT derivation below is written in ambient
mechanism coordinates rather than in the raw neural coordinates. This makes
the derivation applicable to the induced solution path without requiring the
raw parameter vectors to converge or be unique.

\subsection{Constrained mechanism formulation}
Let $\eta$ denote local coordinates for the conditional mechanism
probabilities, and let $\boldsymbol\pi(\eta)$ collect all corresponding
probability vectors. For each $s\in\{+,-\}$, we study a local-minimiser
branch of the abstract constrained problem
\begin{equation*}
    \begin{aligned}
        \min_{\eta}\quad &
        \mathcal{L}_{\mathrm{ED}}(\eta;\widehat{\mathbf q})
        +\lambda a(\eta) \\
        \text{subject to}\quad &
        \mathbf A(\eta)=\mathbf 0,
        \qquad \mathbf B(\eta)\leq\mathbf 0.
    \end{aligned}
\end{equation*}
where the simplex constraints are included in $\mathbf A$ and $\mathbf B$.
For example, each conditional mechanism satisfies
\[
    \sum_j\pi_{i,j}(\eta)=1,
    \qquad
    -\pi_{i,j}(\eta)\leq0.
\]

The neural model in the main text is a smooth feasible map
\[
    \eta=\Phi(\theta),
    \qquad
    \boldsymbol\pi(\Phi(\theta))
    =\operatorname{softmax}\bigl(g(\theta)\bigr).
\]
Thus, the explicit constraints are eliminated in the implementation, but
the induced observational distribution and causal estimand lie in the same
feasible mechanism space. At strictly positive probability vectors, the
softmax map is locally onto the simplex tangent space because its Jacobian
has rank $k-1$. We therefore use the constrained coordinates $\eta$ only
for the local sensitivity calculation. The result transfers to the softmax
model whenever $\Phi$ is locally expressive on the feasible tangent space.
Softmax does not represent boundary probabilities exactly; boundary
solutions are understood as limits of the interior analysis.

\subsection{Local regularity assumptions}
The following assumptions make the constrained KKT system a valid local
description of the solution paths in \Cref{as:smooth}. Fix a regular
reference point $(\eta_0,\lambda_0)$ with $\lambda_0\neq0$, and let
\(\mathcal A=\{r:B_r(\eta_0)=0\}\) be the active inequality set.

First, $\mathcal{L}_{\mathrm{ED}}(\eta;\widehat{\mathbf q})$, $a(\eta)$,
$\mathbf A(\eta)$, and $\mathbf B(\eta)$ are twice continuously
differentiable near $\eta_0$. The active set is fixed under the local
perturbation, and the rows of
\[
    \mathbf{C}_0 =
    \begin{bmatrix}
        \nabla_\eta\mathbf{A}(\eta_0) \\
        \nabla_\eta\mathbf{B}_{\mathcal{A}}(\eta_0)
    \end{bmatrix}
\]
are linearly independent (LICQ). We assume strict complementarity when
needed to ensure that the active set remains fixed.

Finally, the bordered KKT matrix is nonsingular and the second-order
sufficient conditions hold on the critical cone. The former gives a locally
unique differentiable KKT branch by the implicit function theorem, while
the latter ensures that this branch consists of local minimisers. We perform
the calculation at regular points with $\lambda_0\neq0$ and then take a
one-sided limit as $\lambda_0\to0$. Nonsingularity is not required at the
possibly unidentifiable exact-fit problem at $\lambda=0$.

\subsection{KKT perturbation}
Define the Lagrangian
\begin{equation*}
    L(\eta,\mu,\nu;\lambda)
    = \mathcal{L}_{\mathrm{ED}}(\eta;\widehat{\mathbf q})
    +\lambda a(\eta)
    +\mathbf A(\eta)^\top\mu
    +\mathbf B(\eta)^\top\nu.
\end{equation*}
At the reference point, the KKT conditions are
\begin{align*}
    \nabla_\eta L(\eta_0,\mu_0,\nu_{0,\mathcal{A}};\lambda_0) &= \mathbf{0}, \\
    \mathbf{A}(\eta_0) &= \mathbf{0}, \\
    \mathbf{B}_{\mathcal{A}}(\eta_0) &= \mathbf{0}, \\
    \mathbf{B}_{\mathcal{A}^c}(\eta_0) &< \mathbf{0}, \\
    \nu_{0,\mathcal{A}} &\geq \mathbf{0}, \\
    \nu_{0,\mathcal{A}^c} &= \mathbf{0}.
\end{align*}
The third and sixth conditions, together with the active-set definition, imply complementary slackness. The perturbed point is
\[
    (\eta_1,\mu_1,\nu_{1,\mathcal{A}};\lambda_1)
    =
    (\eta_0+\delta\eta,
    \mu_0+\delta\mu,
    \nu_{0,\mathcal{A}}+\delta\nu_{\mathcal{A}};
    \lambda_0+\delta\lambda).
\]
Since the active set is fixed, its stationarity and feasibility conditions are
\begin{align*}
    \nabla_\eta L(\eta_1,\mu_1,\nu_{1,\mathcal{A}};\lambda_1) &= \mathbf{0}, \\
    \mathbf{A}(\eta_1) &= \mathbf{0}, \\
    \mathbf{B}_{\mathcal{A}}(\eta_1) &= \mathbf{0}.
\end{align*}

For compactness, define
\begin{align*}
    \mathbf p_0 &= \mathbf p(\eta_0),
    & a_0 &= a(\eta_0), \\
    \mathbf{A}_\eta &= \nabla_\eta\mathbf{A}(\eta_0), \\
    \mathbf{B}_\eta &= \nabla_\eta\mathbf{B}_{\mathcal{A}}(\eta_0), \\
    \mathbf{C}_0 &=
    \begin{bmatrix}
        \mathbf{A}_\eta \\
        \mathbf{B}_\eta
    \end{bmatrix}, \\
    \mathbf{J}_0 &= \nabla_\eta\mathbf{p}(\eta_0), \\
    \mathbf{H}_0
    &= \nabla_\eta^2 L(\eta_0,\mu_0,\nu_{0,\mathcal{A}};\lambda_0) \\
    &= -2\mathbf{J}_0^\top\mathbf{D}\mathbf{J}_0 \\
    &\quad - 2\sum_{k=1}^K
     \left[\mathbf{D}(\mathbf{p}_0-\widehat{\mathbf q})\right]_k
     \nabla_\eta^2p_k(\eta_0) \\
    &\quad + \lambda_0\nabla_\eta^2a(\eta_0)
    + \sum_r\mu_{0,r}\nabla_\eta^2A_r(\eta_0) \\
    &\quad + \sum_{r\in\mathcal{A}}\nu_{0,r}\nabla_\eta^2B_r(\eta_0).
\end{align*}
Subtracting the reference-point stationarity condition from the perturbed stationarity condition and expanding to first order yields
\begin{align*}
    \mathbf{0}
    ={}& \mathbf{H}_0\delta\eta
    + \mathbf{A}_\eta^\top\delta\mu
    + \mathbf{B}_\eta^\top\delta\nu_{\mathcal{A}}
    + \nabla_\eta a(\eta_0)\delta\lambda
    + \mathcal{O}(\delta\lambda^2).
\end{align*}
Similarly, subtracting the reference feasibility conditions from their perturbed counterparts gives
\begin{align*}
    \mathbf{0}
    &= \mathbf{A}_\eta\delta\eta
    + \mathcal{O}(\delta\lambda^2), \\
    \mathbf{0}
    &= \mathbf{B}_\eta\delta\eta
    + \mathcal{O}(\delta\lambda^2).
\end{align*}
Writing $\delta\xi=(\delta\mu,\delta\nu_{\mathcal{A}})$ and stacking the stationarity and feasibility expansions gives
\begin{equation*}
    \begin{bmatrix}
        \mathbf{H}_0 & \mathbf{C}_0^\top \\
        \mathbf{C}_0 & \mathbf{0}
    \end{bmatrix}
    \begin{bmatrix}
        \delta\eta \\
        \delta\xi
    \end{bmatrix}
    =
    -\delta\lambda
    \begin{bmatrix}
        \nabla_\eta a(\eta_0) \\
        \mathbf{0}
    \end{bmatrix}
    + \mathcal{O}(\delta\lambda^2).
\end{equation*}
Assuming this bordered KKT matrix is nonsingular, define
\[
    \begin{bmatrix}
        \mathbf{w}_\eta \\
        \mathbf{w}_\xi
    \end{bmatrix}
    =
    \begin{bmatrix}
        \mathbf{H}_0 & \mathbf{C}_0^\top \\
        \mathbf{C}_0 & \mathbf{0}
    \end{bmatrix}^{-1}
    \begin{bmatrix}
        \nabla_\eta a(\eta_0) \\
        \mathbf{0}
    \end{bmatrix}.
\]
In particular, the parameter component is
\[
    \mathbf w_\eta
    =
    \begin{bmatrix}
        \mathbf I & \mathbf 0 & \mathbf 0
    \end{bmatrix}
    \begin{bmatrix}
        \mathbf H_0 & \mathbf A_\eta^\top & \mathbf B_\eta^\top \\
        \mathbf A_\eta & \mathbf 0 & \mathbf 0 \\
        \mathbf B_\eta & \mathbf 0 & \mathbf 0
    \end{bmatrix}^{-1}
    \begin{bmatrix}
        \nabla_\eta a(\eta_0) \\
        \mathbf 0 \\
        \mathbf 0
    \end{bmatrix}.
\]
It follows that
\begin{align*}
    \delta\eta &= -\delta\lambda\mathbf{w}_\eta + \mathcal{O}(\delta\lambda^2), \\
    \delta\xi &= -\delta\lambda\mathbf{w}_\xi + \mathcal{O}(\delta\lambda^2).
\end{align*}
Finally, using the leading-order Taylor expansions for $\mathbf p(\eta)$ and
$a(\eta_1)$,
\begin{align*}
    \mathbf{p}(\eta_1)
    &= \mathbf{p}_0-\delta\lambda\mathbf{J}_0\mathbf{w}_\eta
     + \mathcal{O}(\delta\lambda^2), \\
    a(\eta_1)
    &= a_0-\delta\lambda\nabla_\eta a(\eta_0)^\top\mathbf{w}_\eta
     + \mathcal{O}(\delta\lambda^2).
\end{align*}
When $\mathbf{p}_0=\widehat{\mathbf q}$,
\begin{align*}
    \mathcal{L}_{\mathrm{ED}}(\eta_1;\widehat{\mathbf q})
    &= -\bigl(\mathbf{p}(\eta_1)-\widehat{\mathbf q}\bigr)^\top
    \mathbf{D}
    \bigl(\mathbf{p}(\eta_1)-\widehat{\mathbf q}\bigr) \\
    &= -\delta\lambda^2(\mathbf{J}_0\mathbf{w}_\eta)^\top
    \mathbf{D}(\mathbf{J}_0\mathbf{w}_\eta)
    + \mathcal{O}(\delta\lambda^3).
\end{align*}
There are no constant or $\mathcal{O}(\delta\lambda)$ terms at
$\mathbf{p}_0=\widehat{\mathbf q}$. Since $\mathbf{D}$ is conditionally
negative definite, the quadratic leading term is nonnegative.

For the one-sided branches in the main text, take regular reference points
along $\eta_s(\lambda)$ and let $\lambda_0\to0$ from the corresponding
side. Writing
\(
    \mathbf J_s(\lambda_0)
    =\nabla_\eta\mathbf p(\eta_s(\lambda_0))
\)
and letting $\mathbf w_{\eta,s}(\lambda_0)$ denote the KKT response for
that branch, the KKT calculation gives
\[
    \left.\frac{\mathrm d}{\mathrm d\lambda}a(\eta_s(\lambda))
    \right|_{\lambda=\lambda_0}
    =-\nabla_\eta a(\eta_s(\lambda_0))^\top
      \mathbf w_{\eta,s}(\lambda_0),
\]
and
\[
    \left.\frac{\mathrm d}{\mathrm d\lambda}\mathbf p(\eta_s(\lambda))
    \right|_{\lambda=\lambda_0}
    =-\mathbf J_s(\lambda_0)\mathbf w_{\eta,s}(\lambda_0).
\]
Under the one-sided smoothness assumption in the main text, these induced
quantities have one-sided limits at zero. Writing
\[
    c_s
    =\left.\frac{\mathrm d}{\mathrm d\lambda}
      a(\theta_s(\lambda))\right|_{\lambda=0},
    \qquad
    \mathbf v_s
    =\left.\frac{\mathrm d}{\mathrm d\lambda}
      \mathbf p_{\theta_s(\lambda)}\right|_{\lambda=0},
\]
the preceding expansions become
\[
    a(\theta_s(\lambda))
    =a(\theta_s(0))+c_s\lambda+O(\lambda^2)
\]
and
\[
    \mathcal L_{\mathrm{ED}}
    (\theta_s(\lambda);\widehat{\mathbf q})
    =-\lambda^2\mathbf v_s^\top\mathbf D\mathbf v_s
    +O(|\lambda|^3).
\]
Since $\mathbf v_s$ has zero sum, conditional negative definiteness of
$\mathbf D$ gives $-\mathbf v_s^\top\mathbf D\mathbf v_s\geq0$. This is
the local extrapolation result used in the main text. The KKT calculation
does not by itself establish that the branch is globally optimal; the
global-minimiser assumption is what is needed for the bound-recovery
corollary.

\section{Analytical Bounds}\label{app:mediated_bounds}
In this section, we give analytical bounds for the Mediated graph, $Z \rightarrow T \leftrightarrow M \rightarrow Y$. Here $ \leftrightarrow$ denotes mutual confounding between $T$ and $M$. All of $Z, T, M, Y$ are binary, there is no restriction on the confounding variable. This graph structure is beyond the assumptions required by CausalOptim. 

Let the ATE from intervening on node $A$ with target node $B$ be $\textrm{ATE}_{A \rightarrow B}$. 

\begin{theorem}
Let $\delta_Y = \textrm{Pr}(Y=1|M=1) - \textrm{Pr}(Y=1|M=0)$. Then $\textrm{ATE}_{T \rightarrow Y} = \delta_Y \textrm{ATE}_{T \rightarrow M}$.
\end{theorem}
\begin{proof} 
    Because the effect of $T$ on $Y$ is fully mediated by $M$, and the effect of $M$ on $Y$ is unconfounded, the mediation formula gives 
    \begin{align*}
        \Pr(Y=1 \mid \operatorname{do}(T=t)) = \\
        \sum_{m \in \{0,1\}} \Pr(Y=1 \mid M=m) \Pr(M=m \mid \operatorname{do}(T=t)). 
    \end{align*}
    Taking the difference between $t=1$ and $t=0$ yields 
    \begin{equation*}
    \begin{aligned}
        \operatorname{ATE}_{T \rightarrow Y}
        &= \delta_Y \Bigl[
            \Pr\bigl(M=1 \mid \operatorname{do}(T=1)\bigr) \\
        &\qquad\quad
            - \Pr\bigl(M=1 \mid \operatorname{do}(T=0)\bigr)
        \Bigr] \\
        &= \delta_Y \operatorname{ATE}_{T \rightarrow M}.
    \end{aligned}
    \end{equation*}
    where the final equality follows because $M$ is binary. 
\end{proof}
Thus, in order to compute bounds of $\textrm{ATE}_{T \rightarrow Y}$, we only need the bounds of $\operatorname{ATE}_{T \rightarrow M}$. This is exactly the binary IV case which is known \cite{balke1997bounds}. 
\section{Synthetic Data Generation} \label{app:data_generation}

Given a causal directed acyclic graph (DAG), we construct a structural
equation model (SEM) and generate its nodes in topological order. Let
\(\mathcal{V}\) denote the set of endogenous nodes in the graph and
\(\mathcal{U}\) the set of exogenous variables. For each node
\(X_v\), \(v\in\mathcal{V}\), we distinguish between its endogenous parents
\[
    \operatorname{pa}_{\mathcal V}(v)
    \subseteq \mathcal V
\]
and its exogenous parents
\[
    \operatorname{pa}_{\mathcal U}(v)
    \subseteq \mathcal U.
\]
Endogenous parents correspond to other variables represented explicitly in
the causal graph, whereas exogenous parents represent latent sources of
dependence, such as unobserved confounders.

The exogenous variables and node-specific noise terms are sampled as
\begin{align*}
    U_j &\sim \operatorname{Uniform}(0,1),
    && j\in\mathcal U,\\
    \varepsilon_v &\sim \operatorname{Uniform}(0,1),
    && v\in\mathcal V,
\end{align*}
with the \(U_j\)'s and \(\varepsilon_v\)'s mutually independent unless the
graph explicitly specifies a shared exogenous parent. Each endogenous node is
then generated according to the structural assignment
\begin{equation*}
    X_v
    =
    f_v\left(
        X_{\operatorname{pa}_{\mathcal V}(v)},
        U_{\operatorname{pa}_{\mathcal U}(v)},
        \varepsilon_v
    \right),
    \qquad v\in\mathcal V,
\end{equation*}
where \(X_{\operatorname{pa}_{\mathcal V}(v)}\) and
\(U_{\operatorname{pa}_{\mathcal U}(v)}\) denote the collections of
endogenous and exogenous parent values, respectively. A shared exogenous
parent may enter the structural assignments of several nodes, thereby
inducing latent dependence between them.

\subsection{Random structural functions}

The nonlinear structural mechanisms are sampled once when a dataset instance
is created. They are represented by random Fourier functions mapping
\([0,1]\) to \((0,1)\). For a scalar input \(x\), a function is defined as
\begin{align*}
    g(x)
    &=
    b+
    \sum_{k=1}^{K}
    \left[
        a_k\cos\!\left(2\pi(k+\tfrac{1}{2})x\right)
        +
        c_k\sin\!\left(2\pi(k+\tfrac{1}{2})x\right)
    \right], \\
    r(x) &= \operatorname{sigmoid}\!\left(\alpha g(x)\right).
\end{align*}
where we have \(K=10\) Fourier components and
\(\alpha=2\). The coefficients are independently sampled as
\[
    b\sim\mathcal N(0,1),\qquad
    a_k,c_k\sim\mathcal N\!\left(0,k^{-2\delta}\right),
    \qquad \delta=1.5.
\]
Thus, the standard deviation of the coefficient at frequency \(k\) is
\(k^{-\delta}\). This decay suppresses high-frequency components and produces
smooth but nonlinear functions. Half-integer frequencies are used to avoid
forcing the function to have equal values at the two boundaries of the unit
interval.

A separate random function is drawn for each relevant configuration of the
discrete parents. For example, with binary \(Z\), the treatment probability is
\begin{equation*}
    p_D(z,u)=
    \begin{cases}
        r_{D,0}(u), & z=0,\\
        r_{D,1}(u), & z=1,
    \end{cases}
\end{equation*}
and the binary treatment is generated by inverse-transform sampling,
\begin{equation*}
    D=\mathbf{1}\!\left\{\varepsilon_D<p_D(Z,U)\right\}.
\end{equation*}
For a multi-valued parent, such as a ternary instrument, one independent
function is constructed for each parent state and the function corresponding
to the realised state is selected.

For a categorical node with \(k>2\) states, the generator constructs
\(k\) random scores \(s_0,\ldots,s_{k-1}\). In the ternary-outcome
experiments, it first obtains a treatment-specific scalar representation,
\[
    h(D,U)=
    \begin{cases}
        r_{h,0}(U), & D=0,\\
        r_{h,1}(U), & D=1,
    \end{cases}
\]
and evaluates the category scores at this value. The conditional
probabilities are
\begin{equation*}
    \pi_j(D,U)
    =
    \frac{\exp\{s_j(h(D,U))\}}
         {\sum_{\ell=0}^{k-1}\exp\{s_\ell(h(D,U))\}},
    \qquad j=0,\ldots,k-1.
\end{equation*}
The realized category is then obtained using the cumulative probabilities:
\begin{equation*}
    Y
    =
    \min\left\{
        j:
        \varepsilon_Y
        \leq
        \sum_{\ell=0}^{j}\pi_\ell(D,U)
    \right\}.
\end{equation*}

The Fourier coefficients determine the structural mechanisms and remain fixed
for all observations in one dataset. Changing the dataset seed draws a new
collection of functions and therefore a new SEM. Conditional on these fixed
mechanisms, individual observations differ because the root variables and
exogenous noise variables are resampled.

\subsection{Observational samples}

An observational dataset of size \(b\) is obtained by repeatedly sampling the
root and noise variables and then evaluating all structural assignments in
topological order. 

Since the problem is only partially identified, we do not use the ATE of the SEM used to generate the dataset. There are many other SEMs with identical observational distributions. Rather, we use the observational distribution to compute bounds on the ATE. This is done by monte-carlo sampling for $\mathbf{q}$ over a large number of observations, then $\mathbf{q}$ is used in the analytic solver for the true bounds. 

\subsection{Graph details} 
For each graph, the $U$ node is unobserved, acting as unobserved confounding. In the Binary IV, Proxy and Mediated graphs, all observed variables are binary. In the Ternary IV graph, the treatment remains binary (so ATE is easily calculated) while node $Z$ and node $Y$ are both ternary.
The datasets are generated using a random seeds from 0 to 9, with the exception of the Mediated dataset. Many datasets from this graph gave nearly 0 ATE. Instead, for this graph, the first 10 datasets to give nontrivial ATE values were used. The final seeds were $\{1, 3, 5, 6, 7, 10, 11, 16, 17\}$.

\section{Analysis of the baselines} \label{app:baseline_eval}
In \Cref{section:baseline}, we found that the baseline methods did not scale with model size. Specifically, the model performance degraded as the neural network widths increased, even with $\mu$P learning rate scaling. In this section, we analyse the cause of this failure. 

BCCE is a GAN architecture, where a generator and discriminator are adversarially optimised. As the NN widths were scaled up, we found an instability where the discriminator becomes too strong and no useful gradient feedback is given to the generator. This can be seen by the high generator error as the with of the NNs becomes large, as seen in \Cref{tab:gan_loss_width} where the model is evaluated on a ternary IV example. The error measured is the largest absolute difference in probability between the predicted and true distributions. This failure mode is well documented and techniques to enhance their stability have been developed \cite{gan_instablitiy}, but is undesirable for the user. As a result of this issue, BCCE struggles on more complex tasks. 

Similarly, NeuralPartialID's performance slightly degrades with larger NN widths. This method is based on a Lagrange multiplier method, shared with \cite{balazadeh2022partial} that solves the minimax optimisation problem,
\[
\min_{\theta}\max_{\chi \ge 0}
\operatorname{ATE}_{\theta}
+
\chi\left(
S\left(P_{\theta}, P\right)-\alpha
\right)
\]
where $S\left(P_{\theta},P\right)$ denotes the Sinkhorn distance \cite{sinkhorn_loss} between the predicted and observed distributions. The model parameters $\theta$ are updated by gradient descent, while the Lagrange multiplier ($\chi$) is updated according to
$\chi \leftarrow \chi + \chi_{\textrm{LR}}\left[S\left(P_{\theta},P\right)-\alpha\right]$.
The fit quality threshold $\alpha$ is derived from a pretraining phase where the average treatment effect is not optimised. 

In principle, the exact constrained solution is recovered in the limiting case $\chi\rightarrow\infty$ and $\alpha=0$. In practice, however, the update rate of $\chi$ must be carefully balanced against the convergence rate of the model parameters. If these rates are poorly matched, training may become prohibitively slow or unstable. Moreover, the resulting bounds are highly sensitive to the choice of $\alpha$. In practice, the algorithm converges to $\alpha > 0$ and the constraint is relaxed, which tends to produce bounds that are wider than the theoretically optimal bounds. \Cref{tab:npid_parameters} displays the fitted $\alpha$ and final $\chi$ values for different width models, on the same IV ternary example, where we see $\alpha$ is always around $0.13$ and $\chi$ is unstable and fluctuates as NN width is increased. The desired scaling of $\chi\rightarrow\infty$ does not occur. 

\begin{table}
\centering
\small
\setlength{\tabcolsep}{3pt}
\begin{tabular}{lcc}
\hline
Width & Lower error & Upper error \\
\hline
16   & 0.008 & 0.010 \\
32   & 0.010 & 0.011 \\
64   & 0.009 & 0.014 \\
128  & 0.062 & 0.151 \\
256  & 0.371 & 0.399 \\
512  & 0.572 & 0.780 \\
1024 & 0.855 & 0.931 \\
\hline
\end{tabular}
\caption{Error of the GAN generator from BCCE for different NN widths for the upper bound model and lower bound model. Lower error indicates better fit. The learning rate follows $\mu$P scaling.}
\label{tab:gan_loss_width}
\end{table}
\begin{table}[t]
\centering
\small
\setlength{\tabcolsep}{3pt}
\begin{tabular}{lccc}
\hline
Width & $\alpha$ & $\chi_{\min}$ & $\chi_{\max}$ \\
\hline
16   & 0.131 & 12.48 & 10.55 \\
32   & 0.132 & 22.54 & 24.51 \\
64   & 0.136 & 14.18 & 27.61 \\
128  & 0.136 & 0.21  & 36.56 \\
256  & 0.128 & 7.62  & 20.97 \\
512  & 0.137 & 8.07  & 29.98 \\
1024 & 0.135 & 20.15 & 18.58 \\
\hline
\end{tabular}
\caption{Final parameters of the NeuralPartialID models for different width NNs. $\chi_{\min}$, and $\chi_{\max}$ are the final Lagrange multipliers, for the lower and upper bounds and $\alpha$ is the fit quality threshold. 
}
\label{tab:npid_parameters}
\end{table}

\end{document}